\pdfoutput=1

\documentclass[11pt]{article}
\usepackage[a4paper, total={16cm, 24cm}]{geometry}
\usepackage{microtype}
\usepackage{authblk}
\title{\textbf{Explanation Systems for Approval-Based Multiwinner Voting}}
\date{}
\author[1]{Niclas Boehmer}
\author[1]{Luca Kreisel}
\author[2]{Jannik Peters}
\affil[1]{Hasso Plattner Institute, University of Potsdam, Germany}
\affil[2]{Shanghai University of Finance and Economics, China}

\usepackage{xcolor}
\usepackage{graphicx}
\usepackage{subcaption}
\usepackage{pdflscape}
\usepackage{url}

\usepackage{amsmath}
\allowdisplaybreaks
\usepackage{amssymb}
\usepackage{amsthm}
\usepackage{mathtools}
\usepackage{nicefrac}

\usepackage{thmtools}
\usepackage{thm-restate}
\newtheorem{theorem}{Theorem}[section]

\newtheorem{observation}[theorem]{Observation}
\newtheorem{axiom}{Axiom}
\theoremstyle{definition}
\newtheorem{definition}[theorem]{Definition}
\newtheorem{example}[theorem]{Example}

\newcounter{claim}
\newenvironment{claim}{%
  \refstepcounter{claim}%
  \par\noindent Claim~\theclaim.\space\itshape%
}{\par\medskip}
\newenvironment{claimproof}[1][]{%
  \def\claimproofarg{#1}%
  \par\noindent\textit{%
    \ifx\claimproofarg\empty
      Proof of Claim~\theclaim.%
    \else
      #1.%
    \fi
  }\hspace{\labelsep}\quad\upshape\ignorespaces%
}{\hfill $\blacksquare$\par\medskip}

\usepackage{tikz}
\usetikzlibrary{arrows.meta}
\usetikzlibrary{positioning}
\usetikzlibrary{patterns,shapes.arrows}
\usepackage[utf8]{inputenc}
\usepackage{pgfplots}
\DeclareUnicodeCharacter{2212}{-}
\usepgfplotslibrary{groupplots,dateplot}
\pgfplotsset{compat=newest}

\usepackage{hyperref}
\hypersetup{
    colorlinks=true,
    citecolor=green!50!black,
    linkcolor=red!60!black,
    hypertexnames=false,
}
\usepackage[nameinlink]{cleveref}
\crefname{claim}{claim}{claims}
\Crefname{claim}{Claim}{Claims}
\crefname{axiom}{axiom}{axioms}
\Crefname{axiom}{Axiom}{Axioms}
\crefname{definition}{definition}{definitions}
\Crefname{definition}{Definition}{Definitions}
\crefname{theorem}{theorem}{theorems}
\Crefname{theorem}{Theorem}{Theorems}
\crefname{lemma}{lemma}{lemmas}
\Crefname{lemma}{Lemma}{Lemmas}
\crefname{proposition}{proposition}{propositions}
\Crefname{proposition}{Proposition}{Propositions}
\crefname{corollary}{corollary}{corollaries}
\Crefname{corollary}{Corollary}{Corollaries}
\crefname{observation}{observation}{observations}
\Crefname{observation}{Observation}{Observations}
\crefname{example}{example}{examples}
\Crefname{example}{Example}{Examples}
\crefname{algocf}{algorithm}{algorithms}
\Crefname{algocf}{Algorithm}{Algorithms}

\usepackage{natbib}

\usepackage[ruled,noend,linesnumbered]{algorithm2e}
\SetKwProg{Fn}{Function}{:}{end}
\SetAlFnt{\small}
\SetAlCapFnt{\small}
\SetAlCapNameFnt{\small}
\SetAlCapHSkip{0pt}
\IncMargin{-\parindent}
\usepackage{algpseudocode}

\usepackage[inline]{enumitem}
\usepackage{booktabs}
\usepackage{array}
\usepackage{multirow}
\usepackage{tablefootnote}
\usepackage{listings}

\newcommand{\NN}{\mathbb{N}}
\newcommand{\RR}{\mathbb{R}}

\newcommand{\bigO}{\mathcal{O}}

\newcommand{\lnice}{laminar-coherence}

\newcommand{\app}{\mathcal{P}(C)^V}
\newcommand{\clrstr}{20}

\DeclareMathOperator*{\argmin}{arg\,min}

\begin{document}

\maketitle

\bigskip
{\footnotesize\tableofcontents}

\newpage
\begin{abstract}
In approval-based multiwinner voting, voters express approval preferences over a set of candidates, and the goal is to return a winning committee. This model captures a broad range of subset selection problems under preferences. Prior work has focused on the study of binary proportionality axioms that certify whether a given committee is proportionally representative or not. We take a more fine-grained perspective and initiate the study of explanation systems that quantify how a committee represents the electorate, i.e., how much influence each voter exerts, how this influence is allocated across selected candidates, how each candidate is backed by the voters, and why certain candidates were not chosen.

Building on the notion of priceability, we propose price systems as a framework for such explanations.
A price system assigns each voter an individual budget, which they can spend on selected candidates they approve, and each candidate needs to be purchased at a unit price.
Since many price systems can exist for a given outcome, selecting among them requires care. 
We initiate an axiomatic study of price systems and propose several axioms capturing structural coherence, faithful attribution of influence, and alignment with proportionality. 
On the algorithmic side, we introduce a polynomial-time computable rule in which voters continuously gain and exercise influence and show that it satisfies all jointly satisfiable axioms. Experiments on synthetic and real-world instances indicate that our explanations correlate with established proportionality notions and can recover unequal influence when it is present.
\end{abstract}

\section{Introduction}\label{sec:intro}

Subset selection problems arise in various contexts, including hiring, product recommendation, shortlisting, facility location, resource allocation, and content curation. 
To illustrate the breadth of such problems, consider three scenarios:
\begin{enumerate*}
    \item[(1)] \emph{Deliberation.} During an online or in-person deliberation, participants submit suggestions on an issue. A limited number of suggestions must be selected to be pursued further \citep{DGG+26a}.
    \item[(2)] \emph{Review curation.} A platform selects a handful of reviews from hundreds to display prominently to users browsing a product \citep{TNT11a}.
    \item[(3)] \emph{Participatory budgeting.} Based on citizens' votes, a city decides which projects to fund from a set of proposals \citep{PPS21a}. 
\end{enumerate*}
In such settings, multiple stakeholders hold preferences over which items should be selected. 
These stakeholders may be individuals (e.g., deliberating or voting citizens), but they can also represent other entities: in scenario~(2), for instance, each review can be viewed as a stakeholder that ``prefers'' the selection of similar reviews.

Approval-based multiwinner voting \citep{LaSk22a} provides an established modeling framework for subset selection  under preferences: 
An instance consists of a set $C$ of candidates, a target size $k$, and a set of voters, where each voter $i$ approves a subset $A_i\subseteq C$ of candidates.
The goal is to output a subset $W\subseteq C$ of $k$ candidates, called a \emph{committee}.
A central theme in the active literature on approval-based multiwinner voting is proportional representation \citep{LaSk22a,PeSk20a,ABC+16a}, which, informally speaking, prescribes that $x\%$ of the voters should decide on $x\%$ of the committee.  
This paper aligns with this line of thinking, adopting the ideal that each voter should exert equal influence over the outcome.

In this paper, we broaden the research agenda around \emph{explaining and auditing} a given committee.
Previous work has largely focused on binary axioms that yield a yes-or-no answer to whether a committee satisfies a particular interpretation of proportional representation (see end of \Cref{sub:RW}). 
We aim to move beyond this coarse evaluation and towards finer-grained questions such as:
\begin{itemize}
    \item[(Q1)] How over- or underrepresented is each voter in the committee?
    \item[(Q2)] How did individual voters influence the selected candidates?
    \item[(Q3)] Which voters were important for a given candidate to get selected?
    \item[(Q4)] How can it be justified that some unselected candidate was not selected? 
\end{itemize}
Answering these questions enables us to evaluate and understand the outcome of any selection process, thereby  
improving the process's trustworthiness, even if the inner workings of the process are opaque or unknown. 
The audits our framework enables can both thoroughly justify outcomes that merely 
appear unfair and detect genuine unfairness  
at both the instance and voter-level.
Returning to the three scenarios above: 
In scenario~(1), our approach could clarify why a narrowly supported suggestion is shortlisted while other seemingly more popular suggestions are not. 
 In scenario~(2), it could audit whether displaying only high-rated reviews, despite the presence of many low-rated reviews, is driven by overrepresentation of high ratings, or instead because the issues raised in low-rated reviews are adequately covered by the selected high-rated ones.
  In scenario~(3), it could assess whether voters from a district that received no funded project are genuinely underrepresented or simply exercise equal influence through funded projects in other districts.
\footnote{Regarding scenario (1), such explanations would substantially extend recent work by \citet{DGG+26a}, who analyze whether deliberative summarization processes satisfy a global notion of proportional representation. Regarding scenario (2), note that platforms are prohibited from overrepresenting, or otherwise preferentially featuring, positive reviews \citep{ftc2022featuring}. Lastly, regarding scenario (3), such explanations would provide a principled basis for justifying participatory budgeting outcomes that have already been sought, e.g., by \citet{wieliczka2023green}. While our paper focuses on unit-cost projects, our general approach also extends to non-unit costs.}

To address Questions (Q1)--(Q4), we propose to explain outcomes via \emph{price systems}.
Informally, a price system assigns each voter a (not necessarily identical) budget, while each selected candidate must be purchased at a unit price; voters spend their budget on selected candidates they approve, and any unspent money becomes residual. We require that no unselected candidate can be afforded by the combined residuals of its supporters.\footnote{The general idea of price systems is closely related to the proportionality notion of \emph{priceability} \citep{PeSk20a,PPSS21a} (see end of \Cref{sec:intro}).}
In this framework, budgets quantify how influence is distributed across voters and hence highlight how close a given outcome is to the principle of equal influence (cf.\ Q1),  payments reveal how voters exercise their influence and how each selected candidate is backed (cf.\ Q2 \& Q3), and payments and residuals together justify why unselected candidates were not chosen, specifically, how their supporters spend their budgets~(cf.~Q4).

We initiate the axiomatic and algorithmic study of finding explanation systems for how a given committee represents a given approval profile (see \Cref{sec:explanation_systems} for an overview of our agenda). Note that choosing the ``right'' explanation is inherently ambiguous since many explanations may be consistent with a given committee, e.g.,  if two voters have identical approval sets, their roles are interchangeable.
At the same time, we want to develop general-purpose explanation rules that return an explanation for \emph{any} given committee.
To guide this choice in a principled way, we conduct an axiomatic analysis of price systems (\Cref{sec:axioms}) centered around the following three desiderata\footnote{Most of our axioms serve as sanity checks; satisfying a subset should not be interpreted as certifying an explanation's quality beyond doubt. 
In contrast, aiming for axioms that impose strong, universal restrictions proved extremely challenging, given the subjective, nuanced, and ambiguous nature inherent to price systems.}:
\begin{itemize}
    \item \emph{Structural coherence (\Cref{str:coherance}):} We capture the coherence of a price system in terms of respecting structural symmetries in the instance.
    \item \emph{Faithfulness axioms (\Cref{sec:faith}):}  We describe conditions that align the price system with the intended meaning of budgets, payments, and residuals.
    \item \emph{Alignment with proportionality (\Cref{sec:algin}):} We relate budget distributions to existing notions of proportional representation.
\end{itemize}

In \Cref{sec:alg}, we propose our main explanation rule, which we term \texttt{Continuous Phragm{\'e}n}.
The key idea is that voters continuously accumulate budget and spend it on candidates they approve, as long as doing so does not destabilize the resulting price system. We show that \texttt{Continuous Phragm{\'e}n} satisfies all our axioms except for a monotonicity axiom. This is not accidental: we prove an impossibility result showing that monotonicity is incompatible with three of our other axioms. Additionally, we demonstrate that approaches built directly on the existing notion of priceability are insufficient for our goals, as they fail many of our axioms. We also discuss a simple baseline method called \texttt{Equal Split}, which assumes that every selected candidate is equally backed by each of its supporters. While this produces easily interpretable price systems, we show that it fails to fully capture the intended meanings of payments and budgets described above.

In \Cref{sec:exp}, we empirically relate our explanations to the established proportionality axiom EJR+ (\Cref{sec:exp:ejr}), observing that disproportionality according to EJR+ is reflected through low minimum budgets of our computed price systems. Moreover, we demonstrate that our methods can meaningfully recover unequal influence in the selection process at the voter level~(\Cref{sec:exp:power}).

\paragraph{Related Work}\label{sub:RW}
We first discuss the relation of our work to the literature on proportional representation and then further related work on explainability and fairness. 
Our work connects and contributes to the rich literature on proportionality in approval-based multiwinner voting  \citep{BrPe23a,LaSk22a,ABC+16a,MPS23a}.
Although previous research suggests that binary proportionality axioms may lose discriminative power in practice \citep{BGP26a,BFNK19a}, quantitative proportionality measures remain rare \citep{BBC+24a,Skow21a, BeFr26a}.  
Like prior work on proportionality, our explanation systems can certify the proportionality of an outcome. 
However, our approach is substantially more fine-grained. 
We aim to quantify representation at both the instance and voter levels and to examine how voters’ influence is exercised across the selected candidates.
Closest to our work is the binary proportionality notion of priceability, which uses price systems as certificates for the proportionality of an outcome \citep{PeSk20a,PPSS21a}. While we draw inspiration from this approach, we show in \Cref{sec:app_price,sub:price} that the price systems from priceability are not suitable as explanations in our sense.  
Another related work is by \citet{HKMS24a}, who study a game-theoretic model of committee selection via so-called budgeting games. While these budgeting games also lead to price systems, the focus of \citeauthor{HKMS24a} lies in characterizing desirable committees as the equilibria of budgeting games, whereas our approach is applicable to any committee.

More broadly, our work conceptually fits within broader efforts on explainable AI and, specifically, explainable mechanism design \citep{SSK22a}. 
One line of work explains single-winner outcomes as the unique survivor after repeatedly applying a computed set of natural axioms \citep{BoEn20a,NBE22a,BEH22a,CaEn16a,PPPZ20a}. 
A second line studies the margin of victory \citep{BBFN21a,BFJ+24b,Xia12a}, quantifying how close a candidate is to being (un)selected under a given rule in various voting settings. Related to the margin of victory is also the concept of minimal supports, used by \citet{CGM26a} to provide minimal explanations for why the winner of a tournament won.
A third line quantifies voter influence as the probability that a voter or coalition is pivotal in an election under some voting rule generated according to some distribution \citep{Xia23a,FKKN11a,Xia23b}.
However, unlike our work, all three previous lines of work do not provide us with an explanation for how an outcome represents a profile.

Finally, our work relates to fairness in subset selection problems such as text summarization \citep{DBLP:conf/www/ShandilyaGG18,DBLP:conf/www/KeswaniC21} and hiring \citep{DBLP:conf/fat/RaghavanBKL20}. These works typically operationalize fairness through predefined sensitive attributes, whereas our approach captures fairness semantically via preferences.

\section{Preliminaries and Explanation Systems}\label{sec:explanation_systems}

\subsection{General Definitions}
  Let $V$ be a set of $n$ \emph{voters} and $C$ a set of \emph{candidates}. Each voter $i\in V$ is equipped with an \emph{approval set} $A_i \subseteq C$ containing all candidates approved by voter $i$. An \emph{approval profile} $A\coloneqq (A_i)_{i\in V}$ contains each voter's approval set.     Note that $\app$ is the set of all possible approval profiles. 
  For two approval profiles $A\in \app$ and $A'\in \mathcal{P}(C')^{V'}$ with $V\cap V'=\emptyset$, we let $A+A' \in \mathcal{P}(C \cup C')^{V \cup V'}$ denote the concatenation of $A$ and $A'$. 
  We say that an approval profile is \emph{unanimous} if $A_i=A_j$ for all $i,j\in V$ and call a candidate $c \in C$ \emph{unanimously approved} if $c \in A_i$ for all voters $i \in V$. 
  For $c\in C$, we write $V[c]\coloneqq \{i \in V \mid c \in A_i\}$ to denote the \emph{supporters} of candidate $c$ and refer to $|V[c]|$ as $c$'s \emph{approval score}.\footnote{To simplify exposition, we assume in the following that $V[c]\neq \emptyset$ for all $c\in C$ and that $A_i\neq \emptyset$ for all $i\in V$.} Candidates with maximum approval score are \emph{approval winners}.  We refer to non-empty subsets of candidates as \emph{committees}.\footnote{Note that we do not constrain the committee size; our approach creates explanations for any possible given committee.} For an approval profile $A\in \app$, a committee $W\subseteq C$ is called \emph{Pareto optimal} if there is no committee $W'\subseteq C$ with $\lvert W'\rvert = \lvert W\rvert$ such that $|A_i\cap W|\leq|A_i\cap W'|$ for all $i\in V$ and $|A_i\cap W|<|A_i\cap W'|$ for some $i\in V$.
  For an approval profile $A\in \app$, a set of voters $V'\subseteq V$, and set of candidates $C'\subseteq C$, we write $A|_{V',C'}$ to denote the restriction of $A$ onto $V'$ and $C'$, i.e., $A|_{V',C'}\coloneqq(A_i \cap C')_{i\in V'}$. 
  For two integers $x,y\in \mathbb{N}$, we write $[x,y]\coloneqq\{x,x+1,\dots,y\}$ as well as $[x] \coloneqq [1,x]$. 

A special class of approval profiles is that of laminar profiles (see \Cref{fig:lp-sub1}):
\begin{definition}[\citep{PeSk20a}\protect\footnote{Our definition differs slightly, as  \citet{PeSk20a} define such profiles with respect to a fixed committee size. }]\label{def:laminar}
    We call an approval profile $A\in \app$ \emph{laminar} if
    \begin{enumerate*}[label=(\roman*)]
      \item the profile is unanimous, or 
      \item there exists a unanimously approved candidate $c \in C$ and $A|_{V,C \setminus \{c\}}$ is  laminar, or 
      \item the profile $A$ can be written as the union of two laminar profiles with disjoint sets of candidates and voters, i.e.,  there exist partitions $V_1  \uplus V_2=V$ and $C_1  \uplus C_2=C$ so that $A=A|_{V_1,C_1}+A|_{V_2,C_2}$ and both $A|_{V_1,C_1}$ and $A|_{V_2,C_2}$ are laminar.
    \end{enumerate*}    
\end{definition}

\subsection{Price Systems and Explanation Rules: Formal Definitions} \label{ps:def}

  The central concept of our paper is that of a price system. We interpret price systems as explanation systems that are defined with respect to a given approval profile and committee. 
  
  \begin{definition}[Price System and Residual Stability]\label{def:pricesystem}
      A \emph{price system} $\mathbf{ps} = (V,W, p)$ for an approval profile $A\in \app$ and committee $W\subseteq C$ contains the voters $V$, committee $W$, and a \emph{payment function} $p\colon (V \times W) \cup \{(i,i)\mid i\in V\}  \to \mathbb{R}_{\ge 0}$. For each voter $i\in V$ and candidate $c\in W$, $p(i,c)$ specifies the non-negative spending of $i$ towards $c$ and $p(i,i)$ the voter's remaining budget, where:
  \begin{itemize}
      \item  For each voter $i\in V$, it holds that $p(i,c)=0$ for all $c\in W\setminus A_i$, i.e., every voter only spends money on candidates in $W$ that they approve.
      \item For $c\in W$,  we require $\sum_{i \in V} p(i,c)=1$, i.e., $c$ is bought at a \emph{price} of $1$.
  \end{itemize}

    Following the priceability definition by \citet{PeSk20a}, we say that a price system $\mathbf{ps} = (V,W, p)$ for an approval profile $A\in \app$ and committee $W\subseteq C$ is \emph{residual-stable} if  for all $c \in C\setminus W$ it holds that $\sum_{i \in V[c]} p(i,i)\leq 1$.
  \end{definition}

  Under our assumption that each candidate is approved by at least one voter, a residual-stable price system exists for any committee. For instance, set $p(i,c) = \frac{1}{\lvert V[c]\rvert}$ for each $c \in W$ and $i \in V[c]$ and $p(i,i)=0$ for each $i\in V$. 
  For a voter $i\in V$, we call $b_i \coloneqq p(i, i) + \sum_{c \in W} p(i, c)$ the \emph{budget} of $i$, $r_i\coloneqq p(i, i)$ the  \emph{residual} of $i$, and $p_i \coloneqq \sum_{c\in W} p(i,c)$ the \emph{total payment} of $i$; note that $b_i=p_i+r_i$.  We call a price system $\mathbf{ps}=(V,W, p)$ \emph{budget-uniform} if all voters have the same budget, i.e., $b_i=b_j$ for all $i,j\in V$. See \Cref{fig:lp} for examples of budget-uniform, residual-stable price systems.

  Given an approval profile $A\in \app$ and committee $W\subseteq C$, an \emph{explanation rule} $\mathfrak{E}$ maps each pair $(A, W)$ to a set of price systems $\mathfrak{E}(A,W)$.  An explanation rule is \emph{resolute} if $|\mathfrak{E}(A,W)|=1$ for every pair of approval profile $A\in \app$ and committee $W\subseteq C$. For notational convenience, when $\mathfrak{E}$ is resolute, we will sometimes write $\mathfrak{E}(A,W)=\mathbf{ps}$ instead of $\mathfrak{E}(A,W)=\{\mathbf{ps}\}$.
  We will define some axioms on price systems (e.g., 1-stability) and will say that an explanation rule satisfies an axiom if all price systems returned by the explanation rule are guaranteed to satisfy the axiom. Throughout the paper, we will interpret price systems as concrete explanations for a fixed approval profile and committee $(A,W)$, while explanation rules map each $(A,W)$ to its price system(s). 

\subsection{Explanation Systems: Intuition and Discussion}\label{sec:expl-intu}
Suppose we are given an approval profile $A\in\app$ and a committee $W\subseteq C$ selected by some algorithm or a human decision-maker. We want to understand how and in what way the selected candidates are backed by and represent the preferences of the voters (see Q1-Q4 in \Cref{sec:intro}). 
For this, we seek to construct an explanation of how the committee relates to the approval profile. Following the ideal of proportional representation, our explanations will be grounded in the notion of influence: for each voter $i$, an explanation should indicate (i) how much influence $i$ can be seen as having over the selected committee, and (ii) how this influence is exercised across candidates.

We are interested in finding price systems that address these two points through their three components (notably, as discussed later, doing so turns out to be an intricate task):
\begin{description}
    \item[Budget] The budget $b_i$ of voter $i$ should quantify the \emph{influence} that $i$ could, in principle, exert on the outcome. It can be interpreted as the voter's \emph{power}. Note that the budgets are not fixed externally, but derived from the selected price system.
    \item[Payments] Payments should quantify how voters exercise their influence. The payment $p(i,c)$ should  represent voter $i$'s contribution to the selection of candidate $c$, while the total payment $p_i$ measures $i$'s overall exercised influence.
    \item[Residual] The residual $r_i$ of voter $i$ should represent $i$'s \emph{unexercised influence}. Residuals are essential to account for heterogeneous approval sets: for instance, a voter who approves no candidate in $W$ need not be \emph{powerless}; rather, the candidates they approve may simply lack sufficient support from other voters. Residual stability enforces that for any unselected candidate $c \notin W$ the supporters of $c$ do not have sufficient influence to strictly afford $c$.
\end{description}

A well-designed price system 
implementing these intuitions enables several forms of analysis, which provide answers to our questions Q1-Q4 from \Cref{sec:intro}.
\emph{Power distribution:} The budgets $(b_i)_{i\in V}$ provide a concise view of how influence is distributed across voters and which voters are under- or overrepresented in the committee (cf.\ Q1). In particular, (approximate) budget-uniformity certifies that the committee satisfies (approximate) proportional representation. \emph{Attribution of support:} For a given voter $i$ the payments $(p(i,c))_{c\in W}$ describe how $i$'s influence is attributed across the selected candidates. This reveals which candidates are selected due to the voters' support~(cf.~Q2). \emph{Justifying~selected~candidates:} For each  candidate $c\in W$, the incoming payments $(p(i,c))_{i\in V}$ constitute an explanation for why the candidate is bought and because of whose support the candidate got selected. Moreover, the magnitude of payments indicates how critical each supporter is to the candidate's selection (cf.\ Q3). \emph{Justifying unselected candidates:} For each  unselected candidate $c^*\notin W$, residual stability together with the payments $(p(i,\cdot))_{i\in V[c^*]}$ show how the influence of the supporters of $c^*$ is distributed in the committee and help explain why these candidates were selected instead of $c^*$  (cf.\ Q4).

\subsubsection{Relation to Proportionality and Majoritarianism.}
Note that the explanations we seek are inherently non-majoritarian.
In our framework, a voter’s budget must be split across the selected candidates they approve, rather than being reusable in full for each approved candidate.
By contrast, a majoritarian explanatory approach would assign weights to voters and justify $W$ as the set of candidates with the highest weighted approval scores.

In fact, residual-stable price systems are closely linked to proportionality. Recall the established $\alpha$-PJR+ axiom for $\alpha \ge 1$ \citep{BrPe23a}:
  A committee $W$ satisfies $\alpha$\nobreakdash-PJR+ if for every $\ell \in [\lvert W \rvert]$ and group $V' \subseteq V$ of voters with $\lvert V'\rvert \ge \alpha \cdot \ell\frac{n}{\lvert W \rvert}$ it holds that
  $\left\lvert \bigcup_{i \in V'} (A_i \cap W) \right\rvert \ge \ell  \text{ or } \bigcap_{i \in V'} A_i \subseteq W$.
It is known that any committee admitting a budget-uniform residual-stable price system satisfies $1$-PJR+ \citep{BrPe23a}. Even further, the minimum voter budget is directly linked to PJR+; inspecting $\min_{i\in V} b_i$ in a residual-stable price system yields an explicit proportionality guarantee:
\begin{restatable}{proposition}{MinBudgetPJR}\label{prop:minBudgetPJR}
  Let $\mathbf{ps}=(V,W,p)$ be a residual-stable price system for some approval profile $A\in \app$ and committee $W\subseteq C$.
  For any $\alpha\ge 1$, if $\min_{i\in V} b_i > \frac{1}{\alpha} \cdot \frac{\lvert W\rvert}{n}$, then $W$ satisfies $\alpha$-PJR+.
\end{restatable}
\begin{proof}
  Let $\ell \in [|W|]$. Let $k=|W|$ and let $V' \subseteq V$ be an $\alpha \cdot \ell$-large group of voters. If $\bigcap_{i \in V'} A_i \subseteq W$, we are done. Thus, assume there is $c \in \bigcap_{i \in V'} A_i$ with $c \notin W$. Since $\mathbf{ps}$ is residual-stable, it needs to hold that
  \begin{align}
    \label{prop:minBudgetPJR:eq:stable}
    \sum_{i \in V'} r_i \leq 1.
  \end{align} 

  We can rewrite
  \begin{align*}
    \sum_{i \in V'} r_i =  \sum_{i \in V'} (b_i - \sum_{c' \in W} p(i,c')) > \alpha \cdot \ell\frac{n}{k} \cdot \frac{k}{\alpha \cdot n}  - \sum_{i \in V'} \sum_{c' \in W} p(i,c') \ge \ell  - \left\lvert \bigcup_{i \in V'} A_i \cap W \right\rvert.
  \end{align*}
  With \Cref{prop:minBudgetPJR:eq:stable}, it follows that
  \begin{align*}
    \ell  - \left\lvert \bigcup_{i \in V'} (A_i \cap W) \right\rvert < 1
  \end{align*}
  and with this, $\left\lvert \bigcup_{i \in V'} (A_i \cap W) \right\rvert \ge \ell$, since the number of approvals is integer-valued.
\end{proof}
Conversely, the maximum budget is closely related to the extent of overrepresentation exhibited by (groups of) voters. From the definition of a price system, it follows that the minimum average number of distinct supporters per selected candidate, where the minimum is taken over all subsets of selected candidates, is lower bounded by the inverse of the maximum budget. This quantity coincides with the instance's maximin support (MMS) value \citep{DBLP:conf/aft/CevallosS21,DBLP:journals/mp/FernandezGFB24}. Thereby, a low maximum budget certifies that every subset of selected candidates is backed by a sufficiently large group of voters:
  \begin{observation}
    Let $\mathbf{ps}=(V,W,p)$ be a price system for an approval profile $A\in \app$ and committee $W\subseteq C$. Then, $\min_{\emptyset \subset C^*\subseteq W}\frac{|\bigcup_{c\in C^*} V[c]|}{|C^*|}\geq \frac{1}{\max_{i\in V} b_i}$. 
  \end{observation}

Another perspective on voters’ budgets is to interpret them as weights: any price system induces a weighted approval profile in which voter $i$ has weight $b_i$.
Under this interpretation, we search for budgets such that the given committee $W$ fulfills the ideal of proportional representation on the induced weighted profile.
Notably, from the work of \citet{BrPe23a} it follows that if a price system satisfies residual-stability, then in the induced weighted instance, $W$ satisfies weighted~PJR+.

\subsubsection{What makes a good explanation? }

A given approval profile and committee typically admit many (residual-stable) price systems, with some unavoidable ambiguity depending on which aspects of the profile one wishes to emphasize.
Formalizing what constitutes a ``good'' explanation is therefore an extremely delicate task; unsurprisingly so,  given that even the simpler task of certifying proportional representation has led to a rich literature of proportionality axioms. 

We evaluate price systems through a collection of desiderata, formalized as axioms in \Cref{sec:axioms}.

\begin{itemize}
\item \textbf{Structural coherence.} Price systems should be interpretable and well-behaved. Differences in budgets or payments must reflect meaningful differences, not arbitrary algorithmic choices. We formalize this through equal treatment of equals (\Cref{def:eqtev}) and symmetry (\Cref{def:symmetry}), which require symmetric voters to receive symmetric budgets and payments. We also characterize how price systems should look on laminar profiles (laminar-coherence, \Cref{def:lnice}).
\item \textbf{Faithful representation of influence.}    Budgets, payments, and residuals should serve the intended reading  of influence, enabling the analyses described above.
We begin with \emph{1-stability} (\Cref{def:one_w_costStable}), a satisfiable strengthening of residual-stability yielding more meaningful budgets, payments, and residuals.
  We then add axioms targeting (i) payments (single-winner-payment-responsiveness, \Cref{def:single_winner_payments}),
  (ii) residuals (budget-averaging, \Cref{def:budget_averaging}), and (iii) the impact of deleting approvals on budgets (monotonicity, \Cref{def:monotonicity}).
\item \textbf{Alignment with proportionality.} The distribution of budgets should be indicative of the level of proportionality the committee provides: 
  If $W$ satisfies a strong proportionality notion for $A$, then the induced price system should be budget-uniform (laminar-proportional-uniformity, \Cref{def:lpu}; perfect-coverage-uniformity, \Cref{pcu}; and perfect-symmetry-uniformity, \Cref{psu}). If $W$ exhibits a clear lack of proportionality, then the price system should reflect this through unequal budgets (unproportional-responsiveness, \Cref{def:unprop_responsive}).

  \end{itemize}

Most of our axioms should be understood as \emph{sanity checks}: violating them suggests that a price system is misleading or hard to interpret. Conversely, satisfying any single axiom should not be viewed as certifying that an explanation is unquestionable beyond doubt. 

\subsection{Priceability}\label{sub:price}

Price systems are closely related to the proportionality notion of \emph{priceability} introduced by \citet{PeSk20a}. In our terminology, a committee $W$ is priceable for an approval profile $A$ if there exists a budget-uniform and residual-stable price system for $(A,W)$.\footnote{In \citet{PeSk20a}, voters have fixed unit budgets while there is a single price to buy candidates, which may vary between instances. This is equivalent to our normalization in which each selected candidate has price~$1$: if the price to buy candidates is $\rho>0$, dividing all payments and residuals by $\rho$ yields a price-$1$ system (and conversely).} While priceability has a strong explanatory flavor, it has almost exclusively been studied as a binary proportionality notion and we are not aware of detailed analyses of the price systems output by priceability. For our purposes, however, price systems witnessing priceability are a natural starting point (we address the fact that they do not exist for all committees below).
However, budget-uniform and residual-stable price systems are not guaranteed to yield a meaningful or interpretable explanation: 

\begin{figure}
    \centering
    
    \begin{subfigure}[b]{0.24\textwidth}
        \centering
        \resizebox{\textwidth}{!}{\begin{tikzpicture}
        [yscale=0.4,xscale=0.9,
        voter/.style={anchor=south}]
        
        \foreach \i in {1,...,4}
            \node[voter] at (\i-0.5, -1.25) {$\i$};

        \draw[fill=magenta!\clrstr] (0,0) rectangle (4,1);
        \draw[fill=magenta!\clrstr] (0,1) rectangle (4,2);
        \node at (2,0.5) {$c_{1}$};
        \node at (2,1.5) {$c_{2}$};

        \draw[fill=magenta!\clrstr] (0,2) rectangle (2,3);
        \draw[fill=magenta!\clrstr] (0,3) rectangle (2,4);
        \draw[fill=white!\clrstr] (0,4) rectangle (2,5);
        \node at (1,2.5) {$c_{3}$};
        \node at (1,3.5) {$c_{4}$};
        \node at (1,4.5) {$c_{5}$};

        \draw[fill=white!\clrstr] (2,2) rectangle (4,3);
        \draw[fill=white!\clrstr] (2,3) rectangle (4,4);
        \draw[fill=white!\clrstr] (2,4) rectangle (4,5);
        \node at (3,2.5) {$c_{6}$};
        \node at (3,3.5) {$c_{7}$};
        \node at (3,4.5) {$c_{8}$};
        \end{tikzpicture}}
        \caption{Laminar profile}
        \label{fig:lp-sub1}
    \end{subfigure} \quad 
    \begin{subfigure}[b]{0.33\textwidth}
        \centering
        \begin{tikzpicture}[
        voter/.style={circle, thick, draw=black, minimum size=5mm, inner sep=1pt, font=\scriptsize},
        candidate/.style={rectangle, draw=black, thick, minimum size=5mm, inner sep=1pt, font=\scriptsize},
        payment/.style={-Stealth, thick},
        every node/.style={font=\tiny},
        node distance=0.8cm
        ]
        \node[voter] (v1) {$1$};
        \node[voter, right=of v1] (v2) {$2$};
        \node[voter, right=of v2] (v3) {$3$};
        \node[voter, right=0.4cm of v3] (v4) {$4$};
        \node[candidate, below=0.6cm of v1] (c1) {$c_3$};
        \node[candidate, right=of c1] (c2) {$c_4$};
        \node[candidate, right=of c2] (c3) {$c_1$};
        \node[candidate, right=0.4cm of c3] (c4) {$c_2$};
        \draw[payment] (v1) -- (c1) node[midway, left] {$0.7$};
        \draw[payment] (v1) -- (c2) node[pos=0.1, above right, inner sep=0.1pt] {$0.1$};
        \draw[payment] (v2) -- (c1) node[pos=0.1, above left, inner sep=0.1pt] {$0.3$};
        \draw[payment] (v2) -- (c2) node[midway, right] {$0.9$};
        \draw[payment] (v2) -- (c3) node[midway, right] {$0.05$};
        \draw[payment] (v3) -- (c3) node[midway, right] {$0.95$};
        \draw[payment] (v4) -- (c4) node[midway, right] {$1$};
        \draw[payment, loop above, looseness=10, min distance=4mm] (v1) to node[above] {$0.45$} (v1);
        \draw[payment, loop above, looseness=10, min distance=4mm] (v2) to node[above] {$0$} (v2);
        \draw[payment, loop above, looseness=10, min distance=4mm] (v3) to node[above] {$0.3$} (v3);
        \draw[payment, loop above, looseness=10, min distance=4mm] (v4) to node[above] {$0.25$} (v4);
        \end{tikzpicture}
        \caption{Price system 1}
        \label{fig:lp-sub2}
    \end{subfigure}
    \hspace{0.25cm}
    \begin{subfigure}[b]{0.33\textwidth}
        \centering
        \begin{tikzpicture}[
        voter/.style={circle, thick, draw=black, minimum size=5mm, inner sep=1pt, font=\scriptsize},
        candidate/.style={rectangle, draw=black, thick, minimum size=5mm, inner sep=1pt, font=\scriptsize},
        payment/.style={-Stealth, thick},
        every node/.style={font=\tiny},
        node distance=0.8cm
        ]
        \node[voter] (v1) {$1$};
        \node[voter, right=of v1] (v2) {$2$};
        \node[voter, right=of v2] (v3) {$3$};
        \node[voter, right=of v3] (v4) {$4$};
        \node[candidate, below=0.6cm of v1] (c1) {$c_3$};
        \node[candidate, right=of c1] (c2) {$c_4$};
        \node[candidate, right=of c2] (c3) {$c_1$};
        \node[candidate, right=of c3] (c4) {$c_2$};
        \draw[payment] (v1) -- (c1) node[midway, left] {$0.5$};
        \draw[payment] (v1) -- (c2) node[pos=0.1, above right, inner sep=0.1pt] {$0.5$};
        \draw[payment] (v2) -- (c1) node[pos=0.1, above left, inner sep=0.1pt] {$0.5$};
        \draw[payment] (v2) -- (c2) node[midway, right] {$0.5$};
        \draw[payment] (v3) -- (c3) node[midway, left] {$0.5$};
        \draw[payment] (v3) -- (c4) node[pos=0.1, above right, inner sep=0.1pt] {$0.5$};
        \draw[payment] (v4) -- (c3) node[pos=0.1, above left, inner sep=0.1pt] {$0.5$};
        \draw[payment] (v4) -- (c4) node[midway, right] {$0.5$};
        \draw[payment, loop above, looseness=10, min distance=4mm] (v1) to node[above] {$0$} (v1);
        \draw[payment, loop above, looseness=10, min distance=4mm] (v2) to node[above] {$0$} (v2);
        \draw[payment, loop above, looseness=10, min distance=4mm] (v3) to node[above] {$0$} (v3);
        \draw[payment, loop above, looseness=10, min distance=4mm] (v4) to node[above] {$0$} (v4);
        \end{tikzpicture}
        \caption{Price system 2}
        \label{fig:lp-sub4}
    \end{subfigure}
     \caption{Laminar profile and  budget-uniform, residual-stable price systems.
In \Cref{fig:lp-sub1}, voters are placed on a line and each candidate is shown as a box containing exactly the voters who approve it; selected candidates are highlighted in red.
Price systems are depicted as directed bipartite graphs with voter circular vertices (top) and candidate squared vertices (bottom); arc labels indicate payments, and voter self-loops indicate residuals.}
    \label{fig:lp}
\end{figure}
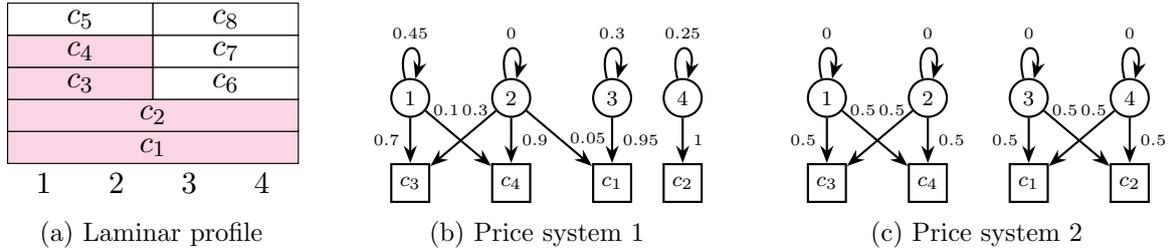

\begin{example}\label{ex:priceability}
Consider the laminar profile and committee in \Cref{fig:lp-sub1}.
\Cref{fig:lp-sub2,fig:lp-sub4} display two budget-uniform, residual-stable price systems for this instance, any of which could be returned as a priceability witness.
The system in \Cref{fig:lp-sub2} is particularly counterintuitive: although voters $1$ and $2$ (and likewise voters $3$ and $4$) have identical approval sets, their residuals and payments differ substantially.
Thereby, the system offers little insight into why the selected candidates were chosen and which voters are responsible for their selection.
\Cref{fig:lp-sub4} treats identical voters symmetrically; yet, we argue that in this instance, a uniform budget assignment fails to reflect the underlying structure.
If all voters had the same influence, it is unclear why one should select both $c_3$ and $c_4$, rather than one of $c_3$, $c_4$, and $c_5$, and one of $c_6$, $c_7$, and $c_8$.
Further, the fact that multiple qualitatively different price systems can witness priceability is problematic.
One could impose additional tie-breaking criteria, e.g., minimizing budgets, to rule out extreme witnesses such as \Cref{fig:lp-sub2}, but the issue of misrepresenting the strength and distribution of influence remains.
\end{example}

Priceability naturally only yields explanations for priceable committees. A natural always-satisfiable relaxation of priceability is to drop budget-uniformity, allowing voters to have different budgets, while continuing to impose residual-stability, i.e., search for a residual-stable price system. However, this relaxation leaves substantial freedom in how budgets are allocated. In particular, voters with identical approval sets can be assigned vastly different budgets, creating artificial influence differences that are not justified by the profile.  As a first explanation rule and natural baseline, we consider \texttt{Approximate Priceability} that selects, among all residual-stable price systems, those whose budgets are as equal as possible, i.e., that minimize the total absolute differences between voters’ budgets  $\sum_{i,j\in V} |b_i-b_j|$ (see \Cref{sec:app_price}). 

In Appendix~\ref{app:price}, we discuss an established extension of priceability and propose a new strengthening of priceability. \citet{PPSS21a} introduce \emph{stable priceability}, yet a stable price system might fail to exist already on small profiles. We propose \emph{weak stability}, a relaxation of \citet{PPSS21a}'s notion that only rules out deviations that would lead to improvements for all involved voters if executed. We show that weakly stable price systems may still fail to exist, and that deciding whether a given price system is weakly stable is \textsc{coNP}-complete. In \Cref{sub:1-stability}, we introduce a satisfiable strengthening of residual stability that is arguably far less demanding than both stability and weak stability.

\section{Axioms}\label{sec:axioms}
In this section, we introduce the axioms we impose on our explanation rules, which aim to align them  with our previously described intuitions of what a ``good'' explanation should entail.

\subsection{Structural Coherence} \label{str:coherance}

Since we want to present computed price systems as explanations, we want them to be comprehensible and intuitive.
In particular, differences between voters should reflect meaningful structural differences in the instance, rather than arbitrary tie-breaking choices of the explanation rule.
A first step in this direction is to impose that voters with identical approval sets are treated identically: they should have the same budget and make the same payments to every selected candidate. 
\begin{axiom}[Equal Treatment of Equals]\label{def:eqtev}
A price system $\mathbf{ps}=(V,W, p)$ for an approval profile $A\in \app$ and committee $W\subseteq C$ satisfies \emph{equal treatment of equals} if for any two voters $i,j\in V$ with $A_i = A_j$ it holds that
\begin{enumerate*}[label=(\alph*)]
\item $p(i,c)=p(j,c)$ for all $c \in W$ and
\item $p(i,i)=p(j,j)$.
\end{enumerate*}
\end{axiom}

However, equal treatment of equals alone is insufficient to rule out price systems that make arbitrary decisions not indicative of structural differences:

\begin{restatable}[Brick-Wall Instance]{example}{brickWallInstance}\label{ex:brick_wall}
Let $V=\{x_1,x_2\}\uplus \{y_1,y_2\}\uplus \{i^*\}$, $C=\{c_1,\dots,c_6\}$, and $W=\{c_1,\dots,c_4\}$.
Voters $x_1$ and $x_2$ approve  $\{c_2,c_4,c_6\}$, voters $y_1$ and $y_2$ approve  $\{c_1,c_3,c_5\}$, and voter $i^*$ approves all candidates.
\end{restatable}

In \Cref{ex:brick_wall}, equal treatment of equals forces $x_1$ and $x_2$ to be treated identically, as well as  $y_1$ and $y_2$.
Nevertheless, a price system may still treat the $x$-voters and $y$-voters differently, for instance, by letting $i^*$ contribute predominantly to the costs of $c_1$ and $c_3$. %
This is counterintuitive: the roles of $\{x_1,x_2\}$ and $\{y_1,y_2\}$ and of $\{c_1,c_3\}$ and $\{c_2,c_4\}$ are fully symmetric in the profile and the committee, so there is no structural reason why $i^*$ should contribute to some selected candidates but not others, or why the total payments of the $x$-voters should differ from those of the $y$-voters.

To rule out such artifacts, we require invariance under automorphisms of the instance.

\begin{definition}
A \emph{profile automorphism} for an approval profile $A\in \app$ is a pair of bijections $\sigma\colon V \to V$ and $\pi\colon C \to C$ such that $A_{\sigma(i)}=\pi(A_i) := \{\pi(c) \mid c \in A_i\}$ for every $i\in V$.
For an approval profile $A\in \app$ and committee $W\subseteq C$, an \emph{automorphism} of $(A,W)$ is a profile automorphism $(\sigma,\pi)$ for $A$ that satisfies $W=\pi(W):= \{\pi(c) \mid c \in W\}$.
\end{definition}
For \Cref{ex:brick_wall}, one automorphism for $(A,W)$ has $\sigma(x_i)=y_i$, $\sigma(y_i)=x_i$ for $i \in \{1,2\}$, $\sigma(i^*)=i^*$, $\pi(c_1)=c_2$, $\pi(c_2)=c_1$, $\pi(c_3)=c_4$, $\pi(c_4)=c_3$, $\pi(c_5)=c_6$, and $\pi(c_6)=c_5$. 

Our next axiom requires that a price system respects all such symmetries.
\begin{axiom}[Symmetry]\label{def:symmetry}
A price system $\mathbf{ps}=(V,W, p)$ for an approval profile $A\in \app$ and committee $W\subseteq C$ satisfies \emph{symmetry} if, for every $i\in V$ and $c \in W$ and every automorphism $\sigma\colon V \to V$ and $ \pi\colon C \to C$ of $(A,W)$, it holds that $
p(i, i) = p(\sigma(i), \sigma(i))$ and $p(i, c) = p(\sigma(i), \pi(c))$.

\end{axiom}
In \Cref{ex:brick_wall}, symmetry implies that $x_1,x_2,y_1,y_2$ have the same residual. Further, it enforces that the $x$-voters both pay an equal amount to $c_2$ and $c_4$ (via the automorphism mapping $c_2$ to $c_4$) and that the payment from each $x$-voter to each of $c_2$ and $c_4$ is the same as the payment from each $y$-voter to each of $c_1$ and $c_3$.
It also implies that $i^*$ contributes equally to each selected candidate (e.g., via automorphisms that fix $i^*$ while permuting candidates).

It is easy to see that symmetry implies equal treatment of equals: for any two voters $i$ and $j$ with $A_i=A_j$,  the automorphism that maps every candidate to itself, every voter from $V\setminus \{i,j\}$ to itself, $i$ and $j$ to each other ensures they are treated identically.

\begin{observation}
    Every price system that satisfies symmetry satisfies equal treatment of equals. 
\end{observation}

Finally, we turn to the special case of laminar profiles (cf. \Cref{def:laminar}).
For such profiles, there is a clear benchmark for what an ``ideal'' explanation should look like: selected candidates’ costs should be shared equally among their supporters, and voters with identical approval sets should have identical residual budgets.
Voters $i$ and $j$ from different ``subparties'' (i.e., profiles as in (iii) of \Cref{def:laminar}) intuitively contribute equal influence to any shared selected candidate $c$ they both approve, since their structural role for $c$'s selection is identical, as any alternative candidate that they disagree on has a smaller supporter set. Hence, on such profiles, differences in power between $i$ and $j$ manifest in which candidates are selected among those they disagree about (their subparties' ``own'' candidates), and should be reflected in the voters' total payments.
Unequal payments for shared candidates would either obscure genuine power disparities or introduce non-uniform payments without structural justification.
This leads to the following axiom:

  \begin{axiom}[Laminar-Coherence]\label{def:lnice}
   For a laminar approval profile $A\in \app$ and committee $W\subseteq C$,   a price system $\mathbf{ps}=(V,W, p)$ is \emph{laminar-coherent} if the following holds: 
    \begin{enumerate}
      \item The payment of each candidate is equally split among its supporters, i.e., for all $c\in W$, $p(i,c)=\frac{1}{|V[c]|}$ for all $i\in V[c]$ and $p(i,c)=0$ for all $i\in V\setminus V[c]$.
      \item Voters with the same approval set have the same residual, i.e., for each $i,j\in V$ with $A_i=A_j$, we have $r_i=r_j$. 
    \end{enumerate}
  \end{axiom}
On \Cref{fig:lp-sub1}, laminar-coherence implies, for example, that all four voters pay $0.25$ each to candidates $c_1$ and $c_2$, and that voters $1$ and $2$ each pay $0.5$ towards $c_5$ and $c_6$.

The laminar-coherence property clearly implies equal treatment of equals: the first condition ensures that voters with the same approval set make identical payments, while the second condition ensures their residuals are equal.
  \begin{observation}
      Let $\mathbf{ps}$ be a laminar-coherent price system for a laminar profile. Then, $\mathbf{ps}$ satisfies equal treatment of equals.
  \end{observation}

\subsection{Faithful Representation of Influence}\label{sec:faith}
We present axioms that check whether a price system faithfully reflects our intended notion of influence.
We propose \emph{1-stability} (\Cref{def:one_w_costStable}), discuss its influence on budgets and payments, and introduce additional axioms on payments (single-winner-payment-responsiveness, \Cref{def:single_winner_payments}), and residuals (budget-averaging, \Cref{def:budget_averaging}).  
Additionally, in Appendix~\ref{app:mono}, we introduce monotonicity (\Cref{def:monotonicity}), an axiom that prescribes that a voter's budget should, relatively speaking, not increase after we delete their approval for a selected candidate. We show an impossibility result for monotonicity with three of our other axioms; none of our explanation rules will satisfy it. 

\subsubsection{Strengthening the Stability Criterion}\label{sub:1-stability}

As observed in \Cref{ex:priceability}, imposing residual-stability alone still permits counterintuitive price systems.
In this section, we focus on one specific source of this problem: the high degree of freedom in  assigning and balancing between residuals and payments it allows. To this end, we propose a satisfiable strengthening of residual-stability. 
Consider the following example: 

\begin{example}\label{ex:rs}
    Let $V=[2q]$, $C=\{c_1,c_2\}$, and $W=\{c_2\}$ where $c_1$ is approved by the voters in $V\setminus\{1\}$ and $c_2$ is approved by $[q]$. Let $p$ be the payment function with $p(i,i)=0$ and $p(i,c_2)=\frac{1}{q}$ for all $i\in [q]$ and $p(j,j)=\frac{1}{q}$ for all $j\in [q+1,2q]$. Then, $\mathbf{ps}=(V,W,p)$ is residual-stable.
\end{example}

There is a clear tension between our interpretation of budgets as voters’ influence and the price system $\mathbf{ps}$ in \Cref{ex:rs}. If all voters had an equal influence on the outcome, there would be no reason to select $c_2$ instead of $c_1$, which is the unique approval winner and approved by all supporters of $c_2$ except voter $1$. The reason why residual-stability still allows for this is that it does not impose that voters spend their budgets effectively. While stronger notions such as stability and weak stability were designed to address this issue, they do not always admit feasible price systems (cf. Appendix~\ref{sec:str_price} for an extended discussion).
We therefore introduce 1-stability. This axiom rules out the possibility of replacing a selected candidate $c' \in W$ with an unselected candidate $c \in C \setminus W$ in a way that allows deviating voters either to strictly improve their utility while using only their residual budget, or to maintain their utility while strictly reducing their total payment.
\begin{axiom}[1-Stability]\label{def:one_w_costStable}
  A price system $\mathbf{ps}=(V,W, p)$ for an approval profile $A\in \app$ and committee $W\subseteq C$ is \emph{1-stable} if it is residual-stable and there are no candidates $c\in C\setminus W$ and $c'\in W$ such that
  $
  \sum_{i\in V[c] \setminus V[c']} r_i + \sum_{i\in V[c']\cap  V[c] }  p(i,c') >1.
  $
\end{axiom}

In \Cref{ex:rs}, 1-stability rules out budget-uniform price systems: either voter $1$ needs to pay a larger share of candidate $c_2$'s cost than the other supporters, or all voters in $[q+1,2q]$ together must have a total budget of at most $1/q$ (otherwise, they form a violating coalition for $c_1$ together with the voters in $[q]\setminus \{1\}$).
1-stability is not only easy to verify but also always satisfiable. In particular, any residual-stable price system can be made 1-stable by setting all residuals to zero.

\begin{restatable}{observation}{OneStablePriceSystem}\label{obs:one_stable_price_system}
  For every approval profile $A\in \app$ and committee $W$, there exists a price system that is $1$-stable and one can check in polynomial time whether a given price system is 1-stable.
\end{restatable}
\begin{proof}
  Verifying whether a given price system is 1-stable can clearly be done in polynomial-time, as it suffices to check for all pairs $c \in C \setminus W$ and $c'\in W$ whether   $\sum_{i\in V[c] \setminus V[c']} r_i + \sum_{i\in V[c']\cap  V[c] }  p(i,c') \leq 1$.

  Consider the price system constructed by splitting the cost of each selected candidate $c\in W$ equally across the supporters, i.e., $p(i,c)=\frac{1}{|V[c]|}$ for all $i\in V[c]$.
  By setting the residuals $r_i=0$ for each voter $i\in V$, we obtain a 1-stable price system, as for all $c \in C \setminus W$ and $c'\in W$, we have
  \[
      \sum_{i\in V[c] \setminus V[c']} r_i + \sum_{i\in V[c']\cap  V[c] }  p(i,c') \leq
\sum_{i\in V[c']} p(i,c')  =  1.
  \]
\end{proof}

\subsubsection{Payments.}
Besides budgets, 1-stability also imposes constraints on payments; it moves them closer to their intended role of indicating which voters were crucial for getting a candidate selected:
\begin{example}\label{ex:1stab_payments}
    Let $V=\{1,2,3\}$, $C=\{c_1,c_2\}$, and $W=\{c_2\}$. Candidate $c_1$ is approved by voters $\{1,2\}$, and $c_2$ is approved by voters $\{2,3\}$. 
    The price system with $p(1,1)=\frac{2}{3}$, $p(2,2)=0$, $p(3,3)=\frac{1}{3}$, $p(2,c_2)=\frac{2}{3}$, and $p(3,c_2)=\frac{1}{3}$ is residual-stable and budget-uniform. However, it misrepresents the instance structure and who is critical for $c_2$'s selection: voter $2$, who approves both candidates, pays more for $c_2$ than voter $3$, who approves only $c_2$. 
    Note that the given price system violates 1-stability. In fact, 1-stability ensures that $p(3,c_2)\geq p(2,c_2)$ in any budget-uniform price system (since otherwise $p(2,c_2)>\frac{1}{2}$ and thus $r_1+p(2,c_2)>2\cdot \frac{1}{2}=1$, rendering $V[c_1]$ a violating coalition).
\end{example}
However, 1-stability alone does not ensure that payments meaningfully reflect different levels of influence. To illustrate, consider \Cref{ex:rs} where the only selected candidate $c_2$ is not an approval winner, i.e., the outcome is suboptimal in a way that suggests unequal influence. Intuitively, among the supporters of $c_2$, those who also approve some approval winner should pay less for $c_2$: these voters are indifferent between this suboptimal outcome and choosing an approval winner that would be selected under equal influence. Thus, as they are not the cause for this unfair selection, they should bear less of the cost of $c_2$. Nevertheless, there exist 1-stable price systems that ignore this distinction, for instance, the system that splits the cost of $c_2$ equally among its supporters and assigns zero residual to everyone. We therefore introduce the following axiom:
\begin{axiom}[Single-Winner-Payment-Responsiveness]\label{def:single_winner_payments}
A price system  $\mathbf{ps}=(V,W,p)$ for an approval profile $A\in \app$ and committee $W=\{c\}$ where $c$ is not an approval winner in $A$ satisfies \emph{single-winner-payment-responsiveness}, if for each pair of voters $i,j\in V[c]$ where $i$ does not approve an approval winner and $j$ approves some approval winner, we have 
$
p(i,c)>p(j,c).
$
\end{axiom}

In Appendix~\ref{sec:app-payments}, we discuss additional criteria for payments. Overall, identifying payment requirements that are broadly desirable and compatible with other desiderata is challenging: strengthening these requirements can easily introduce counterintuitive behavior from other perspectives.

\subsubsection{Residuals.}

Recall that while payments capture the influence voters actively exert, residuals represent unused influence that voters still hold but were unable to use. This may occur, for instance, because all candidates a voter approves are already in the committee, or because their unelected approved candidates lack support by other voters.
Residuals are essential to our framework. If all residuals were set to zero, differences in budgets would no longer provide a reliable signal of unequal voter influence, and thus of potential unfairness. In particular, there would be no way to differentiate between voters who have low payments (and thus a low budget) because they are underrepresented and voters who have low payments, but have no additional ``claim'' to be represented because they are unable to use their additional influence.
This line of reasoning implies that residuals should be used to offset budget inequalities that do not reflect unjustified disparities in influence, but instead arise because a voter cannot legitimately claim additional representation (for instance, because no unelected candidate they approve has sufficiently broad support). This intuition motivates the following axiom, which requires that whenever a voter has a budget below the average payment, this is because allocating them additional residual would violate 1-stability, i.e., it would enable a justified claim to further representation:

\begin{axiom}[Budget-Averaging]\label{def:budget_averaging}
A price system $\mathbf{ps}=(V,W, p)$ for an approval profile $A\in \app$ and committee $W\subseteq C$ satisfies \emph{budget-averaging} if the budget of every voter $i\in V$ with $b_i<\frac{\lvert W\rvert}{n}$ is maximal with respect to 1-stability. Formally, for every $\varepsilon>0$, the price system $\mathbf{ps}'=(V,W, p')$ derived from $\mathbf{ps}$ by setting $r'_i=r_i+\varepsilon$ and keeping everything else unchanged violates 1-stability.
\end{axiom}

\subsection{Alignment with Proportionality} \label{sec:algin}

Since budgets should reflect voters’ influence, unequal budgets should indicate unequal influence and thus a violation of proportional representation. Accordingly, we align the conclusions of our price systems with established proportionality notions. We consider two directions.

\subsubsection{Proportionality Implies Equal Budgets}\label{prop-equal-budget}

When a committee satisfies a strong notion of proportional representation, we want the computed price system to be budget-uniform; otherwise, budget differences would falsely suggest an unfair distribution of influence.

Because proportional representation admits many formalizations, some care is required in operationalizing this requirement. In particular, we only want to enforce budget uniformity in settings where the proportionality of the committee is beyond reasonable doubt (otherwise we would wrongly restrict our price systems). Many mainstream proportionality notions do not meet this standard, for example, even some core-stable committees may plausibly be regarded as unfair \citep{PPSS21a}.
We therefore focus on three settings in which proportionality is unambiguous.

Our first axiom deals with laminar profiles.
In particular, as argued by \citet{PeSk20a}, laminar proportional committees fulfill  a strong notion of proportionality: disagreements between ``subparties'' (i.e., profiles as in (iii) of \Cref{def:laminar}) are fairly resolved by selecting candidates in proportion to their sizes.
\begin{axiom}[Laminar-Proportional-Uniformity]\label{def:lpu}
     For a laminar profile $A\in \app$, a committee $W\subseteq C$ is \emph{laminar proportional} if: 
    \begin{enumerate}
        \item All candidates from $W$ are unanimously approved, or   
        \item there is a unanimously approved candidate $c\in C$ such that $W\setminus \{c\}$\footnote{$W$ does not necessarily contain $c$, as we do not require our committees to be efficient, but are mainly concerned with differences of influence between voters. In the original definition by \citet{PeSk20a}, $c\in W$ is required.} is laminar proportional in $A|_{V,C\setminus \{c\}}$, or
        \item there exist partitions $V_1  \uplus V_2 =V$ and $C_1  \uplus C_2=C$ so that $A|_{V,C}=A|_{V_1,C_1}+A|_{V_2,C_2}$, $\frac{|V_1|}{|W\cap C_1|}=\frac{|V_2|}{| W\cap C_2|}$, $W\cap C_1$ is laminar proportional in $A|_{V_1,C_1}$ and $W\cap C_2$ is laminar proportional in $A|_{V_2,C_2}$.
    \end{enumerate}
    
     An explanation rule $\mathfrak{E}$ satisfies \emph{laminar-proportional-uniformity} if each price system returned by $\mathfrak{E}$ for $A\in \app$ and $W\subseteq C$ is budget-uniform  when $A$ is laminar and $W$ is laminar proportional. 
\end{axiom}

Another strong form of proportional representation arises when the committee consists of the $\lvert W\rvert$ candidates with the highest approval scores and each voter approves exactly one of these candidates. 
Such a committee arguably meets the strongest conceivable standard of proportional representation (while also maximizing both utilitarian and egalitarian social welfare).

\begin{axiom}[Perfect-Coverage-Uniformity]\label{pcu}
  For an approval profile $A\in \app$, a committee $W\subseteq C$ provides \emph{perfect coverage}, if (a) $W$ is an exact cover of $V$, i.e., $|A_i\cap W|=1$ for each $i\in V$, and (b) there is no $c \in C\setminus W$ and $c' \in W$ with $|V[c]|>|V[c']|$. An explanation rule $\mathfrak{E}$ satisfies \emph{perfect-coverage-uniformity} if each price system returned by $\mathfrak{E}$ for $A\in \app$ and $W\subseteq C$ is budget-uniform when $W$ provides perfect coverage for $A$. 
\end{axiom}

A strong notion of proportional representation also emerges in highly symmetric instances.
For example, in \Cref{ex:brick_wall}, the candidates are split into two symmetric types, and the
committee selects the same number of candidates from each type. In such a situation, there is
no structural reason for voters to claim unequal influence, and the explanation should therefore
be budget-uniform:

\begin{axiom}[Perfect-Symmetry-Uniformity]\label{psu}
  Fix an approval profile $A \in \app$ and a committee $W \subseteq C$.
  We call two candidates $c,c' \in C$ \emph{isomorphic} if there exists a profile automorphism
  $(\sigma,\pi)$ of $A$ with $\pi(c)=c'$.
  Moreover, we say that $c$ and $c'$ have the same \emph{type} if $V[c]=V[c']$.

  An explanation rule $\mathfrak{E}$ satisfies \emph{perfect-symmetry-uniformity} if,
  whenever all candidates are pairwise isomorphic and $W$ contains the same number of
  candidates from each type, every price system returned by $\mathfrak{E}$ is budget-uniform.
\end{axiom}

\subsubsection{Non-Proportionality Implies Unequal Budgets}
Conversely, when a committee clearly violates proportional representation, we want the resulting price systems to reflect this through unequal budgets.
Note that if $W$ violates ($1$-)PJR+, then we have already seen that any residual-stable price system for $W$ cannot be budget-uniform. 

On laminar profiles, proportionality requires that the number of candidates selected from two ``subparties'' (i.e., profiles as in (iii) of \Cref{def:laminar}) should be proportional to their size (see \citet{PeSk20a} and \Cref{def:lpu}). We introduce the notion of $\Delta$-laminar-unproportionality, which quantifies the extent to which some committee $W$ deviates from this ideal: $W$ is $\Delta$-laminar-unproportional if there exist two subparties where one is underrepresented by more than $\Delta$ candidates relative to its proportionality demand.
The associated axiom, $\Delta$-unproportional-responsive, then requires that an explanation rule only returns non-budget-uniform price systems for all $\Delta$-unproportional committees; it is desirable that rules already fulfill this axiom for $\Delta=1$.
\begin{axiom}[$\Delta$-Unproportional-Responsive]\label{def:unprop_responsive}
  Let $\Delta\in \NN$. A committee $W\subseteq C$ for a laminar profile $A\in \app$ is \emph{$\Delta$-laminar-unproportional}, if one of the following holds:
  \begin{enumerate}
    \item The profile with candidates $C$ is unanimous,  $W\subseteq C$, and $\Delta=0$, or 
    \item There is a unanimously approved candidate  $c \in C$ such that $W\setminus \{c\}$ is $\Delta$-laminar-unproportional in the profile $A|_{V, C\setminus \{c\}}$, or
    \item There exist partitions $V=V_1\uplus V_2$ and $C=C_1\uplus C_2$ into two laminar profiles with $A=A|_{V_1,C_1}+A|_{V_2,C_2}$ such that: 
    \begin{enumerate}
        \item  The committee $W\cap C_1$ is $\Delta_1$-laminar-unproportional on $A|_{V_1,C_1}$ and $W \cap C_2$ is $\Delta_2$-laminar-unproportional on $A|_{V_2,C_2}$, and
        \item Let $\Delta^* \in \NN$ be maximum such that
      $
      \lvert W \cap C_1\rvert + \Delta^* <  \lvert W \cap C_2\rvert \cdot  \frac{|V_1|}{|V_2|}
      $
      and $ \lvert(\bigcap_{i\in V_1} A_i )\setminus W\rvert \geq \Delta^*$.
      It holds that $\max\{\Delta_1,\Delta_2,\Delta^*\}\geq \Delta$.
    \end{enumerate}
  \end{enumerate}

An explanation rule $\mathfrak{E}$ is \emph{$\Delta$-unproportional-responsive} if each price system returned by $\mathfrak{E}$ for $A\in \app$ and $W\subseteq C$ is not budget-uniform when $A$ is laminar and $W$ is $\Delta$-laminar-unproportional.
\end{axiom}

The committee in \Cref{fig:lp-sub1} is $1$-laminar-unproportional, but not $2$-laminar-unproportional.
We show that satisfying laminar-coherence (cf. \Cref{def:lnice}) and residual-stability (cf. \Cref{def:pricesystem}) ensures $1$-unproportional-responsiveness:
\begin{restatable}{proposition}{LniceResidualImpliesUnpropResp}\label{prop:lnice-residual-unpropresp}
  An explanation rule that is laminar-coherent and residual-stable satisfies $1$-unproportional-responsiveness.
\end{restatable}
\begin{proof}
Let $\mathfrak{E}$ be a laminar-coherent and residual-stable explanation rule, let $A\in \app$ be a laminar profile, and let $W\subseteq C$ be a $1$-laminar-unproportional committee. Let $V=V_1\uplus V_2$ and $C=C_1\uplus C_2$ be partitions into two laminar profiles satisfying 3) in \Cref{def:unprop_responsive} for $\Delta^*=1$, i.e.,
\begin{equation}\label{prop_lnice_unprop:eqStar}
  |W \cap C_1 |  < |W \cap C_2| \cdot \frac{|V_1|}{|V_2|} - 1.
\end{equation}
Note that $\Delta^*\geq 1$ also implies that there exists a candidate $c^* \in (\bigcap_{i\in V_1} A_i) \setminus W$, so $V_1 \subseteq V[c^*]$.

Assume for the sake of contradiction that a price system $\mathbf{ps}\in\mathfrak{E}(A,W)$ is budget-uniform with budget $b^*\in \RR$ for all voters.
Note that \lnice{} implies that each voter $i\in V_1 \cup V_2$ pays the same amount $p^*\coloneqq \sum_{c \in W \setminus (C_1 \cup C_2)} p(i,c)$ for the candidates approved by all (or none) of the voters in $V_1 \cup V_2$.

Since the candidates in $C_2$ can only be funded by voters in $V_2$, we have that
$$
b^* \geq p^* + \frac{|W \cap C_2|}{|V_2|}.
$$
We can lower bound the total residual of the voters in $V_1$ by
\begin{align*}
  \sum_{i \in V_1} r_i = \sum_{i \in V_1}( b^* - p_i )\geq |V_1| \cdot b^* - (p^* \cdot |V_1| +  |W \cap C_1 |)\\
  \geq p^* \cdot |V_1| + \frac{|W \cap C_2| \cdot |V_1|}{|V_2|} - (p^* \cdot |V_1| +  |W \cap C_1 |)\\
  =  \frac{|W \cap C_2| \cdot |V_1|}{|V_2|} - |W \cap C_1 | >_{\eqref{prop_lnice_unprop:eqStar}} \frac{|W \cap C_2| \cdot |V_1|}{|V_2|}-(|W \cap C_2| \cdot \frac{|V_1|}{|V_2|} - 1) = 1.
\end{align*}
Since $V_1 \subseteq V[c^*]$ and $c^* \notin W$, this contradicts that $\mathbf{ps}$ is residual-stable, which completes the proof.
\end{proof}

\section{Algorithms}\label{sec:alg}

We present three explanation rules. We start with \texttt{Approximate Priceability}, then discuss \texttt{Equal Split}, and finally present our main explanation rule \texttt{Continuous Phragm\'{e}n} (see \Cref{tab:axiom_small}).  

\subsection{Baseline: Approximate Priceability} \label{sec:app_price}
\texttt{Approximate Priceability} computes a residual-stable price system that minimizes the sum of absolute pairwise differences between voters’ budgets via a linear program. In particular, it returns a budget-uniform residual-stable price system when one exists.
 \texttt{Approximate Priceability} satisfies our three axioms capturing that proportionality implies budget-uniformity, as committees meeting our strong notions of proportionality naturally admit budget-uniform residual-stable price systems (as they are priceable). 
At the same time, the high flexibility in distributing payments and budgets that priceability allows for (cf.\ \Cref{ex:priceability}) means that \texttt{Approximate Priceability} can output equal budgets even in instances where voters arguably exert very unequal influence. Consequently, it violates $\Delta$-unproportional-responsiveness for arbitrary $\Delta$, and our structural coherence axioms. 
\begin{restatable}{theorem}{priceabilityGuarantees}
\texttt{Approximate Priceability} satisfies \begin{enumerate*}[label=(\alph*)]
    
    \item laminar-proportional-uniformity,
    \item perfect-coverage-uniformity, and
    \item perfect-symmetry-uniformity.
\end{enumerate*}

\texttt{Approximate Priceability} violates
\begin{enumerate*}[label=(\alph*)]
    \item equal treatment of equals,
    \item 1-stability,
    \item single-winner-payment-responsiveness, and
    \item $\Delta$-unproportional-responsiveness for any $\Delta \in \NN$.
\end{enumerate*}
   
\end{restatable}

\subsection{A First Attempt: Equal Split}\label{alg:es}

We present \texttt{Equal Split}, a first attempt towards more regularized and interpretable price systems that first computes a payment scheme and then assigns residuals.
For the payments, \texttt{Equal Split} treats each selected candidate independently and splits its unit cost equally among its supporters, resulting in a 1-stable price system. While this choice is not fully aligned with our intended interpretation of budgets and payments (see below), it yields simple and easily accessible price systems.
To determine voters’ residuals, we follow the principle that residuals should be granted unless there is a compelling reason not to. For our residual distribution, we call a voter \emph{blocked} if increasing their residual would violate 1-stability, and \emph{active} otherwise. We repeatedly increase the residuals of all active voters whose current budgets are minimal among the active voters in a continuous fashion. We continue until either no active voters remain, or further increasing any active voter’s residual would increase the maximum budget in the instance. This procedure is formalized as \textsc{Distribute Residual} in \Cref{alg:cont-phragmen}; we discuss it in more detail in the next section.

It turns out that \texttt{Equal Split} performs well with respect to many of our axioms:
\begin{restatable}{theorem}{eqSplitGuarantees}\label{thm:eqSplitGuarantees}
\texttt{Equal Split} satisfies \begin{enumerate*}[label=(\alph*)]
    \item symmetry,
    \item \lnice{},
    \item 1-stability,
    \item budget-averaging,
    \item laminar-proportional-uniformity,
    \item perfect-coverage-uniformity, and
    \item 1-unproportional-responsiveness.
\end{enumerate*}
\texttt{Equal Split} violates \begin{enumerate*}[label=(\alph*)]
    \item single-winner-payment-responsiveness, and
    \item perfect-symmetry-uniformity.
\end{enumerate*}
\end{restatable}

Nevertheless, \texttt{Equal Split} can produce price systems that fail to reflect the structure of the instance. First, because \texttt{Equal Split} charges every supporter of a candidate the same amount, its payments do not distinguish voters who were critical for the inclusion of a candidate from those who were not. As a result, \texttt{Equal Split} violates, for instance, single-winner-payment-responsiveness. Second, the resulting budgets can be highly unequal even when the committee satisfies strong proportionality ideals:  
In \Cref{ex:brick_wall}, \texttt{Equal Split} assigns voter $i^*$ a payment of $\nicefrac{1}{3}$ for each selected candidate, yielding a budget of $\nicefrac{4}{3}$. Residual-stability then limits the remaining voters to budgets of at most $\nicefrac{7}{6}$ (i.e., total payments of $\nicefrac{2}{3}$ and residuals of $\nicefrac{1}{2}$). Thus, \texttt{Equal Split} suggests that $i^*$ exerts strictly more influence and is effectively overrepresented, even though $i^*$ merely approves all candidates and hence should not affect which candidates are selected in the first place. Consequently, \texttt{Equal Split} also violates perfect-symmetry-uniformity.

\paragraph{Additional Method.}
A different route is to link the computation of payments and residuals through an optimization approach, thereby balancing both simultaneously. In Appendix~\ref{app:subsec:opt}, we discuss a convex optimization approach with objective function 
$\max  \sum_{i \in V} \log p(i,i) + \sum_{c \in W} \sum_{i \in V[c]} \log p(i,c)$.

\newcommand{\yes}{\ensuremath{\checkmark}}
\newcommand{\no}{\ensuremath{\times}}
\begin{table}[t]
\centering
\caption{Guarantees for
  Symmetry (Sym.,~\Cref{def:symmetry}),
  Single-Winner-Payment-Responsiveness (SW-Pay.,~\Cref{def:single_winner_payments}),
  Perfect-Symmetry-Uniformity (PSU,~\Cref{psu}), and
  $\Delta$-Unproportional-Responsiveness ($\Delta$-Lam.,~\Cref{def:unprop_responsive}).
         Every rule satisfies all omitted axioms apart from Monotonicity (\Cref{def:monotonicity}); and for \texttt{Approximate Priceability}  Equal Treatment of Equals (\Cref{def:eqtev}), 1-Stability (\Cref{def:one_w_costStable}), and Laminar-Coherence (\Cref{def:lnice}).}
\label{tab:axiom_small}
\renewcommand{\arraystretch}{1}
\resizebox{0.7\textwidth}{!}{\begin{tabular}{@{}l
                c
                c
                c
                c@{}}
\toprule
\textbf{Explanation Rule}
  & \textbf{Sym.}
  & \textbf{SW-Pay.}
  & \textbf{PSU}
  & $\boldsymbol{\Delta}$\textbf{-Lam.} \\
\midrule
\texttt{Approximate Priceability} (Sec.~\ref{sec:app_price})
  & \no & \no & \yes & $\Delta{=}\infty$\\
  \texttt{Equal Split} (Sec.~\ref{alg:es})
  & \yes & \no & \no & $\Delta{=}1$ \\
\texttt{Continuous Phragm\'{e}n} (Sec.~\ref{sec:cont-phr})
  & \yes & \yes & \yes & $\Delta{=}1$ \\
\bottomrule
\end{tabular}}
\end{table}

\subsection{Continuous Phragm{\'e}n}\label{sec:cont-phr}

We present \texttt{Continuous Phragm{\'e}n}, our main explanation rule. We give an overview of the rule's general idea and properties in \Cref{sec:ov} and a detailed description in \Cref{sec:det}.

\subsubsection{Overview}\label{sec:ov}
We draw inspiration from the voting rule \emph{sequential Phragm{\'e}n} for selecting a committee of size $k$. Under sequential Phragm{\'e}n, each voter maintains a bank account that is initially set to zero. All voters earn money continuously at the same rate. As soon as the supporters of a candidate $c$ have collectively accumulated one unit of money, candidate $c$ is added to the committee and the accounts of all supporters of $c$ are reset to zero, while all other voters keep their current balances. This is continued until $k$ candidates have been selected. 
We adapt the idea of voters continuously earning money to the construction of price systems. A natural starting point would be to treat $W$ as the candidate set and exercise the \emph{sequential Phragm{\'e}n} rule, letting voters ``purchase'' candidates until all candidates in $W$ have been added. The resulting price system would then consist of the incurred payments together with the remaining account balances as residuals. 

However, there are two problems with this approach.
First, as it ignores unselected candidates, the price system obtained in this way need not be residual-stable or 1-stable. For instance, if $W$ contains a candidate approved by only a single voter, that voter would need a long time to pay for this candidate, while the remaining voters would accumulate substantial residuals, yielding significant violations of 1-stability and residual-stability. 
Second, sequential Phragm{\'e}n constructs spending in a highly discrete manner and is very sensitive to tie-breaking. If a voter approves two candidates that become purchasable at the same time, then a sequential-Phragm{\'e}n-style procedure will need to break ties in favor of one of them, forcing the voter to pay substantially more for this candidate. More generally, this will yield price systems that are highly non-symmetric.

\paragraph{High-Level Idea and Properties.}
To address these two challenges, we propose \texttt{Continuous Phragm{\'e}n}, which differs in the following two key ways. First, unlike sequential Phragm{\'e}n, where voters accumulate money and then spend it once in a lump sum on a single candidate, \texttt{Continuous Phragm{\'e}n} lets voters spend continuously over time. More concretely, voters continuously split their spending across (multiple) selected candidates they approve, and, when they have no remaining affordable spending opportunities, accumulate residual. This continuous funding approach side-steps tie-breaking and thus ensures equal treatment of structurally symmetric candidates, as a supporter's spending can be split equally between them.
Secondly, \texttt{Continuous Phragm{\'e}n} explicitly limits both (i) how much residual a voter can accumulate and (ii) how much a voter can pay towards any candidate to ensure that the resulting price system is 1-stable. Note that restricting payments is necessary because if a voter pays too much for a candidate, this can create a 1-stability violation.

We designed \texttt{Continuous Phragm{\'e}n} to satisfy all of our axioms except monotonicity, which is ruled out by our impossibility result in \Cref{prop:no-expl-rule}: 

\begin{restatable}{theorem}{contPhragmenGuarantees}\label{thm:contPhragmenGuarantees}
\texttt{Continuous Phragm{\'e}n} satisfies \begin{enumerate*}[label=(\alph*)]
    \item symmetry,
    \item \lnice{},
    \item \mbox{1-stability,}
    \item single-winner-payment-responsiveness,
    \item budget-averaging,
    \item laminar-proportional-uniformity,
    \item perfect-coverage-uniformity,
    \item perfect-symmetry-uniformity, and
    \item 1-unproportional-responsiveness.
\end{enumerate*}
\end{restatable}

\subsubsection{Detailed Description}\label{sec:det}
\Cref{alg:cont-phragmen} presents \texttt{Continuous Phragm{\'e}n}. The algorithm maintains
the set $V_{\mathrm{active}}$ of \emph{active} voters who currently accumulate residual or pay for selected candidates;
the set $V_{\mathrm{blocked}}$ of \emph{blocked} voters who are currently prevented from accumulating or spending;
the set $W_{\mathrm{critical}}$ of \emph{critical} candidates  that are difficult to fund under the current constraints (see below);
the set $W_{\mathrm{remaining}}\subseteq W$ of selected candidates that have not yet been fully paid for; and a vector
$(h_c)_{c\in W}\in[0,1]^W$, where $h_c$ denotes the remaining unpaid cost of candidate $c$.
A further central object are the \emph{spending sets} $(C_i)_{i\in V_{\mathrm{active}}}$, which prescribe how active voters split their payments.

Conceptually, the algorithm evolves in continuous time and distributes money until all candidates in $W$ are paid for, i.e., until $h_c=0$ for all $c\in W$. For readability, we present this as a loop advancing time in infinitesimal steps of size $\varepsilon$. The algorithm can be implemented with running time $\bigO(n^3\cdot m^2\cdot \lvert W\rvert^2)$; we refer to Appendix~\ref{app:algos} for details.
In each outer iteration, we perform (i) a \emph{planning phase}, which determines spending sets and updates $V_{\mathrm{active}}$, $V_{\mathrm{blocked}}$, and $W_{\mathrm{critical}}$ until they stabilize, (ii) a \emph{spending phase}, in which voters spend according to the computed spending sets, and (iii) a \emph{pruning phase}, which reduces residuals to eliminate 1-stability violations that may arise from unblocking voters to fund $W_{\mathrm{critical}}$ (explained below).
Once all candidates in $W$ have been fully paid for, we execute a final \emph{residual phase} that allocates additional residual while preserving 1-stability.

\paragraph{Spending phase.}
In Lines~\ref{line:spending-start}--\ref{line:spending-end}, voters accumulate and spend money. For each active voter $i\in V_{\mathrm{active}}$, we maintain the voter's spending set $C_i$ of candidates on which $i$ currently spends. In each infinitesimal step, voter $i$ distributes the incoming amount $\varepsilon$ uniformly across the candidates in $C_i$. If $C_i=\emptyset$, then $i$  adds $\varepsilon$ to their residual $p(i,i)$.
In the spirit of sequential Phragm{\'e}n, the spending sets prioritize spending on candidates that can be funded quickly, i.e., candidates that require small per-supporter contributions given the remaining unpaid costs (see following paragraph).\footnote{A simpler alternative would be to set $C_i=A_i\cap W_{\mathrm{remaining}}$ for all $i$. However, this may introduce unjustified differences in spending for shared candidates between supporters, and violates laminar-proportional-uniformity (cf.\ \Cref{ex:cont-phragmen-ci-alt}). }

\begin{algorithm}[t!]
\caption{Continuous Phragm{\'e}n}
\label{alg:cont-phragmen}
 $V_{\mathrm{active}}\gets V$, $V_{\mathrm{blocked}}\gets \emptyset$, $W_{\mathrm{critical}}\gets \emptyset$, $W_{\mathrm{remaining}}\gets W$, $h_c\gets 1$ for all $c\in W$\;
 $p(i,c) \gets 0, p(i,i) \gets 0$ for all $i\in V$, $c \in W$\;
 \While{$W_{\mathrm{remaining}}\neq \emptyset$: for an infinitesimal $\varepsilon$ step \label{line:outer-while}}{
\Repeat(\tcp*[f]{Planning phase}){$(C_i)_{i\in V_{\mathrm{active}}}, V_{\mathrm{active}},V_{\mathrm{blocked}},W_{\mathrm{critical}}$ remain unchanged}{
  $(C_i)_{i\in V_{\mathrm{active}}}\gets$ \textsc{MoneyFlow}{$(V_{\mathrm{active}}, (h_{c})_{c\in W}, W_{\mathrm{remaining}})$}\label{line:C-fund}\;
 $V_{\mathrm{blocked}}\gets V_{\mathrm{blocked}} \cup $\textsc{Check Blocking}$(V_{\mathrm{active}}$, $W_{\mathrm{critical}}$,$p$, $(C_i)_{i\in V_{\mathrm{active}}})$\label{line:V-blocked}, $V_{\mathrm{active}}\gets V_{\mathrm{active}} \setminus V_{\mathrm{blocked}} $\;
  $W_{\mathrm{unsupp}}\gets \{c\in W_{\mathrm{remaining}}\mid V[c]\cap V_{\mathrm{active}}=\emptyset\}$\label{line:uns}\;
 $W_{\mathrm{critical}}\gets W_{\mathrm{critical}}\cup W_{\mathrm{unsupp}}, V_{\mathrm{active}}\gets V_{\mathrm{active}}\cup \bigcup_{c\in W_{\mathrm{unsupp}}}V[c], V_{\mathrm{blocked}}\gets V_{\mathrm{blocked}}\setminus V_{\mathrm{active}}$\label{line:crit}\;
}
 \For(\tcp*[f]{Spending phase}){each voter $i\in V_{\mathrm{active}}$\label{line:spending-start}}{
 \lIf{$C_i\neq\emptyset$}{ 
    $h_c\gets h_c-\nicefrac{\varepsilon}{|C_i|}$, $p(i,c)\gets p(i,c)+\nicefrac{\varepsilon}{|C_i|}$  for each  $c\in C_i$
  }
  \lElse{$p(i,i)\gets p(i,i)+\varepsilon$\label{line:spending-end}}}
  $t_i \gets 1$ for all $i\in V$ \label{line:pruning-start}  \tcp*[r]{Pruning phase}
 \For{each  $c \in C\setminus W$ and $c'\in W$ with $\delta\coloneqq (\sum_{j\in V[c] \setminus V[c']} p(j,j) + \sum_{j\in V[c']\cap  V[c] }  p(j,c') )-1>0$}{
    $t_i \gets \min(t_i,  (1- \nicefrac{\delta}{\sum_{j\in V[c] \setminus V[c']} p(j,j)}))$ for each $i\in  V[c] \setminus V[c']$\;
 }
 $p(i,i) \gets p(i,i) \cdot t_i$ for each $i\in V$\label{line:pruning-end}\;
 
    $W_{\mathrm{remaining}}\gets \{c\in W\mid h_c>0\}$ and $W_{\mathrm{critical}}\gets \{c\in W_{\mathrm{critical}}\mid h_c>0\}$\label{line:update-remaining}\;
}
\textsc{Distribute Residual}$(p)$ \tcp*[r]{Residual phase}\label{line:dis}
\KwRet $(V,W,p)$\;
\BlankLine
\Fn{\textsc{MoneyFlow}{$(V_{\mathrm{active}}, (h_{c})_{c\in W}, W_{\mathrm{remaining}})$}}{
  $C_i\gets \emptyset$ for all $i\in \{j\in V_{\mathrm{active}}\mid A_{j}\cap W_{\mathrm{remaining}}=\emptyset \}$ and $V^*\gets \{j\in V_{\mathrm{active}}\mid A_{j}\cap W_{\mathrm{remaining}}\neq\emptyset \}$\;
  \While{$V^*\neq \emptyset$}{
$C'\gets \argmin_{c\in W_{\mathrm{remaining}} : V[c]\cap V^* \neq \emptyset} \frac{h_c}{|V[c]\cap V^*|}$\;
\For{$i\in V^*$}{\lIf{$A_i\cap C'\neq \emptyset$}{$C_i\gets A_i\cap C'$, $V^*\gets V^*\setminus \{i\}$\label{line:assign-ci}}}
  }
  \KwRet $(C_i)_{i\in V_{\mathrm{active}}}$\;
}

\BlankLine
\Fn{\textsc{Check Blocking}{$(V_{\mathrm{active}}$, $W_{\mathrm{critical}}$,$p$, $(C_i)_{i\in V_{\mathrm{active}}})$}}{
 $V^* \gets \emptyset$\;
 \For{each unselected candidate $c \in C\setminus W$}{
    \lIf{$\sum_{i\in V[c]}p(i,i)=1$}{ $V^* \gets V^* \cup \{i\in V[c] \cap V_{\mathrm{active}} \mid A_i\cap W_{\mathrm{critical}}=\emptyset\wedge C_i=\emptyset   \}$}
    
    \For{each $c'\in W$ with
  $\sum_{i\in V[c] \setminus V[c']} p(i,i) + \sum_{i\in V[c']\cap  V[c] }  p(i,c') =1
  $}{$V^* \gets V^* \cup (\{i\in (V[c] \setminus V[c']) \cap V_{\mathrm{active}} \mid A_i\cap W_{\mathrm{critical}}=\emptyset\wedge C_i=\emptyset   \}\cup \{i\in V[c']\cap  V[c] \cap V_{\mathrm{active}} \mid A_i\cap W_{\mathrm{critical}}=\emptyset\wedge c' \in C_i\})$}  
}
  \KwRet $V^*$\;
}

\BlankLine
\Fn{\textsc{Distribute Residual}{$(p)$}}{
  $V_{\mathrm{active}}\gets V$; $V_{\mathrm{blocked}}\gets \emptyset$\;
  \While{$V_{\mathrm{active}}\neq \emptyset$ and $\min_{i\in V_{\mathrm{active}}} b_i<\max_{i\in V} b_i$: for an infinitesimal $\varepsilon$ step}{
    $V_{\mathrm{blocked}}\gets V_{\mathrm{blocked}} \cup $\textsc{Check Blocking}$(V_{\mathrm{active}}$, $\emptyset$,$p$, $(\emptyset)_{i\in V_{\mathrm{active}}})$,$V_{\mathrm{active}}\gets V_{\mathrm{active}}\setminus V_{\mathrm{blocked}}$\label{line:block-res}\;\
    $p(i,i)\gets p(i,i)+\varepsilon$ for each $i\in \{j \in V_{\mathrm{active}} \mid b_j = \min_{j'\in V_{\mathrm{active}}} b_{j'} \}$\label{line:distr-res}\;
  }
  \KwRet $p$\;
}
\end{algorithm}

\paragraph{Planning phase.}

The planning phase consists of several steps that are repeated until a fixed point is reached (i.e., until $V_{\mathrm{blocked}}$, $V_{\mathrm{active}}$, $W_{\mathrm{critical}}$, and the spending sets $(C_i)_{i\in V_{\mathrm{active}}}$ no longer change).

Line \ref{line:C-fund}: We compute $(C_i)_{i\in V_{\mathrm{active}}}$ via \textsc{MoneyFlow}. We set $C_i=\emptyset$ whenever $A_i\cap W_{\mathrm{remaining}}=\emptyset$.
Otherwise, \textsc{MoneyFlow} assigns spending sets by iteratively ``covering'' voters: let $V^*\subseteq V_{\mathrm{active}}$ be the set of active voters whose spending sets are not yet assigned. For each $c\in W_{\mathrm{remaining}}$, we compute the time until $c$ would be fully paid if all supporters with an unassigned spending set $V[c]\cap V^*$ funneled their money exclusively to $c$, namely $\frac{h_c}{|V[c]\cap V^*|}$.
Let $C'$ be the set of candidates attaining the minimum ratio (i.e., those candidates that can be bought with minimum contributions by each supporter), and assign each voter $i\in V^*$ with $A_i\cap C'\neq\emptyset$ the set $C_i:=A_i\cap C'$. Remove these voters from $V^*$ and repeat until $V^*=\emptyset$.

Line \ref{line:V-blocked}: In \textsc{Check Blocking}, we block any voter $i$ who would create a 1-stability violation if they were to receive money in the next infinitesimal step and who does not approve a candidate from $W_\mathrm{critical}$ (see below):
if $C_i=\emptyset$, this means that increasing $p(i,i)$ would violate 1-stability; if $C_i\neq\emptyset$, it means that increasing $i$'s payments towards some candidate in $C_i$ would violate 1-stability.

Lines \ref{line:uns} and \ref{line:crit}: A critical issue arises when all supporters of some remaining candidate $c$ become blocked. Then $c$ would never be fully paid for.\footnote{This is unavoidable, for instance, when $W$ is not Pareto optimal and there is an unselected candidate $c' \in C\setminus W$ with $V[c]\subset V[c']$. Resolving this is necessary to ensure that our explanation rule returns a price system for every committee.} 
We therefore declare such a candidate $c$ \emph{critical} by adding it to $W_{\mathrm{critical}}$, and we unblock its supporters by re-adding them to $V_{\mathrm{active}}$. Moreover, in \textsc{Check Blocking}, we do not block voters who approve any currently critical candidate.

We iterate until a fixed point is reached, i.e., until the sets of blocked voters, critical candidates, and spending sets remain unchanged from one iteration to the next. This guarantees that the subsequent spending phase does not introduce any 1-stability violations, except for those handled later in the pruning phase.\footnote{If we executed the loop only once, spending sets would be outdated at the point where money gets spent as they would be computed with respect to a set of active voters that might subsequently get changed.}
The inner loop executes at most $n+\lvert W\rvert$ times before a fixed point is reached, since in each iteration we must either add a new candidate to $W_{\mathrm{critical}}$ or block a previously unblocked voter. Each candidate enters $W_{\mathrm{critical}}$ at most once, and each voter can enter $V_{\mathrm{blocked}}$ at most once, as they are only removed if they approve a critical candidate and can then only be re-added later once this candidate is fully paid for.

\paragraph{Pruning phase.}
Unblocking supporters of a critical candidate $c^*$ in Line \ref{line:crit} can lead to violations of 1-stability when these supporters increase their spending on selected candidates; note that supporters of $c^*$ are ``ignored'' during \textsc{Check Blocking} and are thus free to cause such deviations. We therefore consider all pairs $(c,c')$ with $c\in C\setminus W$ and $c'\in W$ that violate the 1-stability constraint, and compute the corresponding excess $\delta= (\sum_{i\in V[c] \setminus V[c']} p(i,i) + \sum_{i\in V[c']\cap  V[c] }  p(i,c') )-1>0$.
While one could eliminate such a deviation by setting the residuals of all voters in $V[c]\setminus V[c']$ to zero (note that  $\delta\leq \sum_{i\in V[c]\setminus V[c']} p(i,i)$), this is overly aggressive. Instead, we scale down the residuals of the voters in $V[c]\setminus V[c']$ uniformly, reducing each residual by the same fraction
$\nicefrac{\delta}{\sum_{i\in V[c]\setminus V[c']} p(i,i)}$.
If a voter participates in multiple violating pairs, we apply the strongest required scaling (i.e., the maximum reduction factor among the violations involving that voter) once in Line~\ref{line:pruning-end}.

\paragraph{Residual phase.}
In Line~\ref{line:dis}, we allocate additional residual. Following \Cref{alg:es}, we maintain the set $V_{\mathrm{active}}$ of voters for whom increasing $p(i,i)$ would not violate 1-stability. Among these voters, we identify those with minimum current budget and increase their residuals by an infinitesimal amount~$\varepsilon$. We repeat until either (i) $V_{\mathrm{active}}$ becomes empty (so any further increase would violate 1-stability) or (ii) the minimum budget among  voters from $V_{\mathrm{active}}$ reaches the maximum budget in the instance, at which point we stop to prevent excessive residual accumulation.

\section{Experimental Analysis}\label{sec:exp}
We empirically investigate whether our explanations align with established proportionality notions (\Cref{sec:exp:ejr}) and can reveal latent differences in voter influence (\Cref{sec:exp:power}).
We report results on $1000$ profile--committee pairs sampled from the Euclidean-VCR model, a standard model for synthetic, approval elections \citep{DBLP:conf/ijcai/BoehmerFJ0LPRSS24}. In each run, we sample $n,m\in[10,100]$ and a committee $W\subseteq C$ of size $k= \lfloor m/2 \rfloor$ uniformly at random. To generate the approval profile, we embed voters and candidates independently and uniformly in the two-dimensional unit square. For each profile, we sample an approval radius $r\in[0.05,0.3]$ and let the voters approve all candidates within distance~$r$.
In Appendix~\ref{app:exp}, we present additional experiments, among others  results for profiles sampled from the resampling model, as well as for real-world participatory budgeting data from Pabulib, where the results are qualitatively comparable. We also vary the committee size using $k\in\{\lfloor m/8 \rfloor,\,\lfloor m/4 \rfloor\}$. The code for all experiments is available on GitHub: \url{https://github.com/luca-kreisel/Explanation-Systems}. 

\subsection{EJR+ and Budgets}\label{sec:exp:ejr}

We study the empirical relationship between budget distributions and proportionality. To this end, we use \textsc{EJR+}, one of the strongest known always satisfiable proportionality notions:
\begin{definition}[\citep{BrPe23a}]
For an approval profile $A \in \app$ and $\alpha>0$, a committee $W$ satisfies $\alpha$-\textsc{EJR+} if there is no candidate $c \notin W$, group of voters $V' \subseteq V$, and $\ell \in \mathbb{N}$ with $|V'| \geq \alpha \cdot \frac{\ell n}{|W|}$ such that
$
c \in \bigcap_{i \in V'} A_i \text{ and }  |A_i \cap W| < \ell \text{ for all } i \in V'.
$
\end{definition}
For each profile $A\in\app$ and committee $W\subseteq C$, we compute the minimum value $\alpha$ for which $W$ satisfies $\alpha$-EJR+ in~$A$. The larger the values of $\alpha$ the less proportional the committee. 
\Cref{fig:ejr:main} shows scatter plots where each dot corresponds to an instance; the $x$-axis reports the minimum budget as a fraction of $\nicefrac{\lvert W\rvert}{n}$, and the $y$-axis reports the minimum $\alpha$ value for EJR+. Values below $1$ on the $x$-axis correspond to non-uniform budgets and signal under-representation according to our explanation systems, while values above $1$ on the $y$-axis correspond to violations of $1$-EJR+.

We ask whether non-proportionality, as captured by EJR+, results in non-uniform budgets. For \texttt{Continuous Phragm{\'e}n} and  \texttt{Equal Split}\footnote{\label{fn:equalsplit}The similar performance of \texttt{Equal Split} compared to \texttt{Continuous Phragmén} may seem surprising given that it can assign undesirable payments (cf.~\Cref{alg:es}). A possible explanation is that on the tested instances, residuals—which are set in the same manner for both—may play a more important role than payments.  This suggests that 1-stability, which both methods satisfy but \texttt{Approximate Priceability} does not, may be a significant reason for their superior performance.}, we observe a strong association: instances with larger $\alpha$ tend to exhibit smaller minimum-budget fractions and in particular every committee that is unproportional under EJR+ also has non-uniform budget and a minimum budget fraction below $1$ (cf. \Cref{fig:ejr:k-hist-three}). In contrast, for \texttt{Approximate Priceability}, the correlation is noticeably weaker. In particular, there are many instances with $\alpha>1$ for which the method still outputs uniform budgets. 

Note that the converse does not hold in general: satisfying $1$-EJR+ does not imply budget-uniformity. This is expected, since EJR+ captures only one aspect of proportional representation, and there are instances where committees satisfy EJR+ despite exhibiting   inequalities in influence.

\begin{figure}[t]
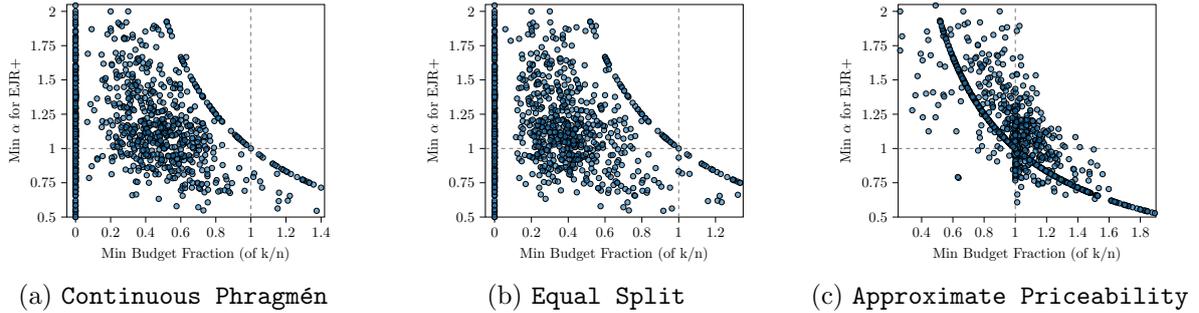

\centering
\begin{subfigure}[b]{0.32\textwidth}
\centering
\resizebox{!}{3.5cm}{%
    \input{figures/tikzplotlib/EUCLIDEAN2DVCRwithk=m0.5CONTPHRAGMENscatter}
}
\caption{\texttt{Continuous Phragmén}}
\end{subfigure}
\hfill
\begin{subfigure}[b]{0.32\textwidth}
\centering
\resizebox{!}{3.5cm}{%
    \input{figures/tikzplotlib/EUCLIDEAN2DVCRwithk=m0.5EQUALSPLITscatter}
}
\caption{\texttt{Equal Split}}
\end{subfigure}
\hfill
\begin{subfigure}[b]{0.32\textwidth}
\centering
\resizebox{!}{3.5cm}{%
    \input{figures/tikzplotlib/EUCLIDEAN2DVCRwithk=m0.5APPROXPRICEABILITYscatter.tex}
}
\caption{\texttt{Approximate Priceability}}
\end{subfigure}
\caption{Scatterplot of minimum budget fraction and $\alpha$-EJR+ approximation for $k=\lfloor m/2 \rfloor$ on profiles sampled from the  Euclidean-VCR model and random committees. Axes are clipped to the 5th–95th percentile range. }
\label{fig:ejr:main}
\end{figure}

\subsection{Recovering Voting Power}\label{sec:exp:power}
One use case of explanation systems is to recover unequal voter influence. We therefore ask:
\emph{Given an approval profile and a committee produced by a proportional rule run on the profile, but where voters were assigned different weights, can our explanation rules recover the underlying voter weights via the assigned budgets?}
This question is highly non-trivial, as the committee alone does not uniquely determine the weights, e.g., any splitting of some voting weight on two voters with an identical approval set would lead to the same outcome for most rules.

We build on Euclidean-VCR instances and generate voter weights in a structured way using a spatial bias: we sample a point $q$ uniformly at random in the two-dimensional unit square. 
For each voter, we compute their weight following a Gaussian decay\footnote{The parameters were chosen such that voters exhibit a meaningful range of weights across the unit square.} as: $\exp(\frac{-d^2}{2\cdot 0.3^2})$, where $d$ is the distance between the voter and $q$. We then compute the committee $W$ using the method of equal shares \citep{joss-abcvoting,PeSk20a} in the resulting weighted instance.

For each instance and each computed price system, we measure the Pearson correlation coefficient (PCC) between the ground-truth voter weights and the budgets output by the explanation rule (run on the profile without information about the weights).
For \texttt{Continuous Phragm{\'e}n}, 11.9\% of instances exhibit a very strong correlation (PCC $> 0.9$), 65.3\% show a strong correlation (PCC $> 0.7$), and 91.1\% show a moderate correlation (PCC $> 0.4$); see \citep{schober2018correlation} for these interpretations of PCC values.
For \texttt{Equal Split}, these figures are 28.1\%, 72.8\%, and 92.9\%, respectively (cf. Footnote~\ref{fn:equalsplit} on its similar performance). For \texttt{Approximate Priceability} they are 3.3\%, 37.3\%, and 79.3\%.
Overall, these results suggest that our explanations, especially those computed by  \texttt{Continuous Phragm{\'e}n} and \texttt{Equal Split}, can meaningfully reconstruct latent unequal influence.

\section{Conclusion}
We initiated the formal study of explanation systems for multiwinner voting. Explanation systems enable a more nuanced analysis of committees and their proportionality
beyond standard binary axioms. In particular, we advocate for using price systems as explanations and introduce \texttt{Continuous Phragm{\'e}n} as an axiomatically grounded explanation rule. Using these price systems, we obtain a principled way to measure influence in an election. Our work opens several directions. First, on the technical side, it would be valuable to develop and experiment with additional explanation rules and axioms. Further, we would be interested in identifying axioms that lend themselves to characterization results. Beyond allowing non-uniform budgets, an interesting extension is to permit non-uniform candidate prices as an additional indicator of misrepresentation. Second, it is natural to extend explanation systems to other collective decision-making settings, such as multiwinner voting with ranked or cardinal preferences (for instance, deliberation platforms typically allow participants to rate comments or to cast up- and down-votes) or participatory budgeting with heterogeneous project costs. Finally, it is promising to apply explanation systems to audit fairness in broader subset selection problems.
As an example, consider extractive text summarization, where a subset of documents from a corpus is selected as a summary. Prior work has highlighted unfairness in established summarization algorithms using group-fairness notions based on sensitive attributes \citep{DBLP:journals/pacmhci/DashSBGGC19,DBLP:conf/www/KeswaniC21,DBLP:conf/www/ShandilyaGG18}. This setting fits our model: we can view each document as a candidate and as a voter approving semantically similar documents. In Appendix~\ref{app:exp:textsum}, as a proof of concept, we apply our framework to audit summaries produced by the extractive LSA algorithm \citep{DBLP:journals/jasis/DeerwesterDLFH90} on a dataset derived from a deliberation hosted on \citet{polis_data_2026}. We find that LSA summaries are not proportionally representative of the topics in the corpus, and that the price systems computed by \texttt{Continuous Phragm{\'e}n} and \texttt{Equal Split} are able to reveal this disproportionality.

\newpage

\addtocontents{toc}{\protect\setcounter{tocdepth}{-1}}
\bibliographystyle{plainnat}

\newpage

\appendix

\section{Strengthening Priceability}\label{app:price}\label{sec:str_price}

We are not the first to observe that priceability can be too unrestrictive in many cases: \citet{PPSS21a} make similar points and introduce the notion of stable priceability, which additionally requires that voters spend their money efficiently. Specifically, they require that no unselected candidate  can be afforded with the supporters either paying slightly less than their residual or maximum payment to any single selected candidate:

\begin{definition}\label{def:s_costStable}
  A price system $\mathbf{ps} = (V,W, p)$ for an approval profile $A\in \app$ and committee $W\subseteq C$ is \emph{stable} if for all $c \in C\setminus W$ it holds that
  $
    \sum_{i \in V[c]} \max\left(r_i,\: \max_{c^* \in W} p(i, c^*) \right) \leq 1.
  $

\end{definition}

A committee is called stable priceable if a budget-uniform, stable price system exists. \citet{PPSS21a} show that stable priceability is a desirable proportionality notion (for instance, it implies core stability), yet also demonstrate that some instances admit no committee satisfying stable priceability. We show that contrary to residual stable price systems, stable price systems need not exist for all committees (even when the price system need not be budget-uniform).
\begin{observation}\label{app:obs:nonExStable}
    There exists an approval profile $A\in \app$ and a committee $W\subseteq C$ such that no price system is stable for $(A, W)$.
\end{observation}
\begin{proof}
    Consider a profile with $V=[2]$ and candidates $C=\{c^*, c_1, c_2\}$, where all voters approve $c^*$, and candidates $c_1$ and $c_2$ are approved only by voters $1$ and $2$, respectively. The committee is $W=\{c_1,c_2\}$. In any price system $\mathbf{ps}$ for $W$, it needs to hold that $p(1,c_1)=1$ and $p(2,c_2)=1$, which however directly implies that stability is violated for candidate $c^*\in C\setminus W$.
    \end{proof}

This motivates us to search for weakenings of the stability condition. Indeed, the stability condition is quite strong: the hypothetical deviation it prohibits—supporters of a candidate pooling their money to additionally buy that candidate—would often not be beneficial for these voters, as it might cause many candidates they approve to no longer be fully funded: while a supporter only pools away money from one currently selected candidate, others might pool away money from other selected candidates the supporter approves.

We propose a weaker version of stability that only prohibits deviations that are actually attractive to voters. Specifically, we ask whether there exists an unselected candidate $c$ that its supporters can afford in a way such that they either increase their utility, or maintain the same utility while increasing their residual after payments are rearranged.
 
\begin{definition}\label{def:w_costStable}
  A price system $\mathbf{ps} = (V,W, p)$ for an approval profile $A\in \app$ and committee $W\subseteq C$ is \emph{weakly stable} if for all $c \in C\setminus W$ it holds that
  \[
    \max_{\substack{X,Y \subseteq V[c] \\ X \cap Y = \emptyset}} \:\max_{\substack{S \subseteq W \\ \forall i \in X,\: |S \cap A_i| \leq 1 \\ \forall i \in Y,\: S \cap A_i = \emptyset}}
    \:\sum_{i \in X } \max_{c^* \in S} p(i, c^*)  +  \sum_{i \in Y }r_i \leq 1.
  \]
\end{definition}
 It is easy to see that for some approval profiles and non-Pareto-optimal committees, no weakly stable price system exists (e.g., consider the example in \Cref{app:obs:nonExStable}). Unfortunately, requiring that the committee be Pareto optimal does not suffice to ensure existence.
\begin{restatable}{proposition}{weakCostAwareNonEx}\label{prop:weakCostAwareNonEx}
Weakly stable price systems may fail to exist, even for Pareto optimal committees, i.e., 
  there is an approval profile $A\in \app$ with Pareto optimal committee $W\subseteq C$ for which no price system is weakly stable.
\end{restatable}
\begin{proof}
  Consider the following instance with $n = 8$ voters $V=[n]$ and $m = 9$ candidates $C$:
  \begin{itemize}
    \item $c_1$ approved by $\{2, 4, 5, 7, 8\}$,
    \item $c_2$ approved by $\{1, 2, 5, 6\}$,
    \item $c_3$ approved by  $\{1, 3, 5, 7\}$,
    \item $c_4$ approved by  $\{2, 6, 7, 8\}$,
    \item $c_5$ approved by $\{1, 3, 4, 5, 6, 7\}$,
    \item $c_6$ approved by  $\{1, 5, 8\}$,
    \item $c_7$ approved by $\{1, 2, 3, 4, 7\}$,
    \item $c_8$ approved by $\{3, 4, 7, 8\}$,
    \item and $c_9$ approved by $\{3, 4, 6, 8\}$.
  \end{itemize}
  The committee is $W = \{c_2, c_5, c_8\}$, one can verify that $W$ is indeed Pareto optimal.

  Suppose for the sake of contradiction that there exists a weakly stable price system $\mathbf{ps}$ for $W$. Without loss of generality, we can assume that $p(i,i)=0$ for all $i\in V$. In the following, we will consider a subset of the constraints that $\mathbf{ps}$ needs to satisfy, and we will show that these lead to a contradiction and are thus infeasible.
  For brevity, we will shorten $p(i,c_\ell)$ by $p_{i,\ell}$ for $i\in V$ and $c_\ell \in C$. Now, since each selected candidate needs to be funded at a price of $1$, we in particular have the following two constraints for $c_2$ and $c_8$.
  \begin{align}
    &\sum_{i\in V[c_2]} p(i,c_2) = p_{1,2}+p_{2,2}+p_{5,2}+p_{6,2}=1, \label{propNonEx:eq1}\\
    &\sum_{i\in V[c_8]} p(i,c_8) = p_{3,8}+p_{4,8}+p_{7,8}+p_{8,8}=1. \label{propNonEx:eq2}
  \end{align}
  Moreover, simplifying the weak stability constraint by removing the residual contribution implies that for all $c \notin W$, the following inequality needs to hold.
  \[
    \max_{\substack{A \subseteq V[c]}} \:\max_{\substack{S \subseteq W \\ \forall i \in A,\: |S \cap A_i| \leq 1}}
    \:\sum_{i \in A } \max_{c^* \in S} p(i, c^*) \leq 1
  \]
  In particular, this implies that for $S=\{c_2,c_8\}$ and
  \begin{align}
    c=c_1,\: A=\{2,4,5,7,8\}: \: & \sum_{i \in A } \max_{c^* \in S} p(i, c^*) = p_{2,2}+p_{4,8}+p_{5,2}+p_{7,8}+p_{8,8}\le 1, \label{propNonEx:eq3}\\
    c=c_3,\: A=\{1,3,5,7\}: \: & \sum_{i \in A } \max_{c^* \in S} p(i, c^*) = p_{1,2}+p_{3,8}+p_{5,2}+p_{7,8}\le 1, \label{propNonEx:eq4}\\
    c=c_4,\: A=\{2,6,7,8\}: \: & \sum_{i \in A } \max_{c^* \in S} p(i, c^*) = p_{2,2}+p_{6,2}+p_{7,8}+p_{8,8}\le 1, \label{propNonEx:eq5}\\
    c=c_6,\: A=\{1,5,8\}: \: & \sum_{i \in A } \max_{c^* \in S} p(i, c^*) = p_{1,2}+p_{5,2}+p_{8,8}\le 1, \label{propNonEx:eq6}\\
    c=c_7,\: A=\{1,2,3,4,7\}: \: & \sum_{i \in A } \max_{c^* \in S} p(i, c^*) = p_{1,2}+p_{2,2}+p_{3,8}+p_{4,8}+p_{7,8}\le 1, \label{propNonEx:eq7}\\
    c=c_9,\: A=\{3,4,6,8\}: \: & \sum_{i \in A } \max_{c^* \in S} p(i, c^*) = p_{3,8}+p_{4,8}+p_{6,2}+p_{8,8}\le 1. \label{propNonEx:eq8}
  \end{align}
  Now, summing \eqref{propNonEx:eq3}--\eqref{propNonEx:eq8}, we get
  $$
  3\cdot(p_{1,2}+p_{2,2}+p_{3,8}+p_{4,8}+p_{5,2})+2p_{6,2}+4\cdot(p_{7,8}+p_{8,8})\leq 6.
  $$
  Rearranging \eqref{propNonEx:eq1} as $p_{1,2}+p_{2,2}+p_{5,2}=1-p_{6,2}$, and \eqref{propNonEx:eq2} as $p_{3,8}+p_{4,8}=1-p_{7,8}-p_{8,8}$, we can rewrite this as
  \begin{align}
    3\cdot(1-p_{6,2}+1-p_{7,8}-p_{8,8})+2p_{6,2}+4\cdot(p_{7,8}+p_{8,8})&=6-p_{6,2}+p_{7,8}+p_{8,8}\leq 6 \notag\\
    \iff p_{7,8}+p_{8,8} &\leq p_{6,2}.\label{propNonEx:eq9}
  \end{align}
  Since subtracting the left-hand side of \eqref{propNonEx:eq8} from the left-hand side of \eqref{propNonEx:eq2} needs to be at least $0$, we have
  \begin{align}
    p_{3,8}+p_{4,8}+p_{7,8}+p_{8,8}-p_{3,8}-p_{4,8}-p_{6,2}-p_{8,8}=p_{7,8}-p_{6,2}\geq0 \iff p_{7,8}\geq p_{6,2}.\label{propNonEx:eq10}
  \end{align}
  Together, it follows from \eqref{propNonEx:eq9} and \eqref{propNonEx:eq10} that
  $$
  p_{7,8}\geq p_{6,2}\geq p_{7,8}+p_{8,8},
  $$
  so $p_{8,8}=0$ and
  \begin{align}
    p_{7,8}= p_{6,2}.\label{propNonEx:eq11}
  \end{align}
  Using \eqref{propNonEx:eq11}, we can reformulate \eqref{propNonEx:eq1} as follows
  \begin{align}
    p_{1,2}+p_{2,2}+p_{5,2}+p_{7,8}=1.\tag{1$'$}\label{propNonEx:eq1pr}
  \end{align}
  Now, as above, subtracting the left-hand side of \eqref{propNonEx:eq4} from the left-hand side of \eqref{propNonEx:eq1pr} gives
  \begin{align}
    p_{1,2}+p_{2,2}+p_{5,2}+p_{7,8}-p_{1,2}-p_{3,8}-p_{5,2}-p_{7,8}=p_{2,2}-p_{3,8}\geq 0 \iff p_{2,2}\geq p_{3,8}. \label{propNonEx:eq12}
  \end{align}
  Using $p_{8,8}=0$ in \eqref{propNonEx:eq2}, we get
  \begin{align}
    p_{3,8}+p_{4,8}+p_{7,8}=1 \iff p_{4,8}=1-p_{3,8}-p_{7,8}. \tag{2$'$}\label{propNonEx:eq2pr}
  \end{align}
  Substituting \eqref{propNonEx:eq2pr} and $p_{8,8}=0$ into \eqref{propNonEx:eq3} gives
  \begin{align}
    p_{2,2}+1-p_{3,8}-p_{7,8}+p_{5,2}+p_{7,8}+p_{8,8}\le 1 \iff p_{2,2}+p_{5,2}\le p_{3,8}. \label{propNonEx:eq13}
  \end{align}
  From \eqref{propNonEx:eq12} and \eqref{propNonEx:eq13}
  $$
  p_{2,2}+p_{5,2}\le p_{3,8}\le p_{2,2},
  $$
  so $p_{5,2}=0$ and
  \begin{align}
    p_{3,8}=p_{2,2}. \label{propNonEx:eq14}
  \end{align}
  Rearranging \eqref{propNonEx:eq1pr} using $p_{5,2}=0$ as $p_{1,2}+p_{2,2}+p_{7,8}=1 \iff p_{1,2}=1-p_{2,2}-p_{7,8}$ and \eqref{propNonEx:eq14}, we can rewrite \eqref{propNonEx:eq7} as
  \begin{align*}
    1-p_{2,2}-p_{7,8}+p_{2,2}+p_{2,2}+p_{4,8}+p_{7,8}\le 1.
  \end{align*}
  Furthermore, substituting $p_{4,8}$ using \eqref{propNonEx:eq2pr}, and applying the identity from \eqref{propNonEx:eq14} gives
  \begin{align}
    1-p_{2,2}-p_{7,8}+p_{2,2}+p_{2,2}+1-p_{2,2}-p_{7,8}+p_{7,8}=2-p_{7,8}\le 1 \iff 1 \le p_{7,8}.\label{propNonEx:eq15}
  \end{align}
  Using $p_{8,8}=0$ and the identity from \eqref{propNonEx:eq11}, we rewrite \eqref{propNonEx:eq5} as
  $$
  p_{2,2}+p_{6,2}+p_{7,8}+p_{8,8}=p_{2,2}+p_{7,8}+p_{7,8}=2p_{7,8}+p_{2,2} \le 1
  \implies p_{7,8} \leq \frac{1}{2},
  $$
  which however contradicts \eqref{propNonEx:eq15} and concludes the proof.
\end{proof}

This means there is no existence guarantee, which leads us to rule out imposing weak stability as a requirement. Note that \Cref{prop:weakCostAwareNonEx} also strengthens \Cref{app:obs:nonExStable}. Moreover, we show that checking whether a given price system is weakly stable is computationally intractable, further limiting the practicality of this solution concept:
\begin{restatable}{proposition}{weakCostAwareStabilityCoNP}\label{prop:weakCostAwareStabilityCoNP}
  For a given approval profile $A\in \app$, committee $W\subseteq C$, and  price system $\mathbf{ps}$, it is coNP-complete to check if
  $\mathbf{ps}$ is weakly stable, even if $W$ is Pareto optimal in $A$.
\end{restatable}
\begin{proof}
  We prove the statement by showing NP-completeness for the complement problem: deciding if there is a deviation as in \Cref{def:w_costStable} for some candidate $c\notin W$. Containment of the complement problem in NP is clear, since verifying that a candidate  $c\notin W$, sets $X,Y \subseteq V[c]$, and $S \subseteq W$ present a deviation can be done in polynomial time. To prove NP-hardness, we reduce from X3C. An instance for X3C consists of a set $U$ of size $3q$ and a family $\mathcal{F}=\{F_1, \dots, F_\ell\}$ of 3-element subsets of $U$. The X3C problem is to decide if there exists an \emph{exact cover} $\mathcal{F}'\subseteq \mathcal{F}$, i.e., every element in $U$ is contained in exactly one set in $\mathcal{F}'$.
  Given an instance of X3C, we construct an instance in polynomial time as follows:
  \begin{itemize}
    \item The election consists of voters $V=U \cup \{i^*, i^\varepsilon \}$ and candidates $C=\{c_{F_1}, \dots, c_{F_\ell}, c^*\}$. Candidate $c^*$ is approved by $U \cup \{i^\varepsilon \}$, and all other candidates $c_{F_a}$ are approved by $F_a \cup \{i^*\}$. The committee is $W=C \setminus \{c^*\}$, it is easy to verify that $W$ is PO since $i^*$ approves exactly the candidates in $W$.
    \item The price system $\mathbf{ps}$ is defined as follows: Voter $i^\varepsilon$ has a residual $r_{i^\varepsilon }=\varepsilon$ for some $0<\varepsilon<\frac{1}{3q}$, all other voters have a residual of zero. Voter $i^*$ pays $p(i^*, c)=\frac{q-1}{q}$ for all $c \in W$. For candidate $c_{F_a}\in W$, every one of the three voters $i\in F_a$ pays $p(i,c_{F_a})=\frac{1}{3q}$.
  \end{itemize}
  \paragraph{$\Rightarrow$}
  Let $\mathcal{F}'$ be an exact cover. Then, there is a deviation to $c^*$ with $S= \{c_{F_a} \mid F_a \in \mathcal{F}'\}\subseteq W$, $X=U\subseteq V[c^*]$, and $Y=\{i^\varepsilon\}\subseteq V[c^*]$. Note that this is a valid deviation, since $i^\varepsilon$ approves no candidates in $W$, and all voters in $U$ approve exactly one candidate in $S$, since $\mathcal{F}'$ is an exact cover.  Every one of the $3q$ voters $i\in U$ contributes $\frac{1}{3q}$ and voter $i^\varepsilon$ contributes their residual of $\varepsilon>0$, so
  $$
  \sum_{i \in X } \max_{c' \in S} p(i, c')  +  \sum_{i \in Y }r_i = 3q\cdot \frac{1}{3q} + \varepsilon > 1.
  $$
  \paragraph{$\Leftarrow$} Suppose there exists a deviation for weak stability. Clearly, this deviation needs to be for candidate $c^* \notin W$, since all other candidates are in $W$. Thus, there exist $X,Y \subseteq V[c^*]$ and a set $S\subseteq W$ satisfying the conditions in \Cref{def:w_costStable} such that
  $$
  \Delta \coloneqq \sum_{i \in X } \max_{c' \in S} p(i, c')  +  \sum_{i \in Y }r_i > 1.
  $$
  In the following, we argue that it needs to hold that $X=U$ and $Y=\{i^\varepsilon\}$, so $X,Y$ is a partition of $V[c^*]$. Recall that each of the $3q$ voters in $U$ has a maximum payment of $\frac{1}{3q}$ and $i^\varepsilon$ pays for no candidate and has a residual of $r_{i^\varepsilon}=\varepsilon<\frac{1}{3q}$. Thus, if $U\nsubseteq X$, it cannot be that $\Delta>1$. Moreover, if $i^\varepsilon \notin Y$, then $\Delta\leq 1$.
  Next, we will show that $\mathcal{F}'=\{F_a \mid c_{F_a} \in S\}\subseteq \mathcal{F}$ is an exact cover. Firstly, every element in $i\in U$ is contained in at most one set in  $\mathcal{F}'$, since $i \in F_a \iff c_{F_a} \in A_i$ and $|A_i \cap S|\leq 1$ by the restriction on $S$ in \Cref{def:w_costStable} as $i\in X$. Secondly, there cannot be an element $p\in U$ that is not contained in any set in  $\mathcal{F}'$, since otherwise $\Delta\leq 1$, as then there are at most $3q$ voters with a positive contribution (as then $A_i\cap S = \emptyset$  for some voter $i\in X$), and each voter can contribute at most $\frac{1}{3q}$.
\end{proof}

\section{Proofs \& Additional Discussion for Section 3}

\subsection{Monotonicity} \label{app:mono}
Suppose we remove the approval of a voter $i^*$ for a selected candidate and examine how budgets change. Interpreting budgets as a measure of influence suggests that $i^*$ should not receive a higher budget after this modification, as their representation has strictly worsened.
Since budgets should mostly be interpreted in relative terms, our monotonicity axiom therefore requires that, relative to other voters, the influence of $i^*$ does not increase.

\begin{axiom}[Monotonicity]\label{def:monotonicity}
  Let $A\in \app$ be an approval profile and $W\subseteq C$ a committee. We  
  let $A'\in \app$ be the approval profile derived from $A$  by removing the approval of a voter $i^*\in V$ for some selected candidate $c^*\in W$ with at least two supporters, i.e., $A'_{i^*}=A_{i^*} \setminus \{c^*\}$ and $V[c^*]\setminus \{i^*\}\neq \emptyset$.
  A resolute explanation rule $\mathfrak{E}$ satisfies \emph{monotonicity}, if it holds for $\mathbf{ps}=\mathfrak{E}(A,W)$ and $\mathbf{ps}'=\mathfrak{E}(A',W)$, and for each $j\in V\setminus \{i^*\}$ that
  $
  b_j > b_{i^*} \implies  b'_j > b'_{i^*} \text{ and } b_j = b_{i^*} \implies  b'_j \geq b'_{i^*}.
  $
\end{axiom}

It will turn out that none of our explanation methods satisfy monotonicity. However, this is not accidental: the following impossibility result shows that monotonicity is fundamentally incompatible with 1-stability, budget-averaging, and perfect-coverage-uniformity (we introduce perfect-coverage-uniformity in \Cref{prop-equal-budget} as an extremely mild connection between budgets and proportionality): 
\begin{restatable}{proposition}{NoExplRule}\label{prop:no-expl-rule}
There is no explanation rule that is 1-stable, budget-averaging, monotone, and perfect-coverage-uniform.
\end{restatable}
\begin{proof}
Consider any resolute explanation rule $\mathfrak{E}$ that is 1-stable, budget-averaging, and perfect-coverage-uniform.
Let $V=[5]$, $C=\{c_1,c_2,c'_1,c'_2,c'_3\} \cup \{c^*_i \mid i \in V\}$, and $W=\{c_1,c_2\}$. In profile $A$, candidate $c_1$ is approved by $\{1,2,3\}$, and $c_2$ is approved by $\{4,5\}$. Candidate $c'_\ell$ is approved by $\{\ell, 5\}$, and candidate $c^*_\ell$ only by voter $\ell$ . Note that $W$ provides perfect coverage, so by perfect-coverage-uniformity, it needs to hold that $b_4=b_1=b_2=b_3$ in $\mathbf{ps}=\mathfrak{E}(A,W)$.

Consider the profile $A'$ derived by deleting $4$'s approval for $c_2$, and let $\mathbf{ps}'=\mathfrak{E}(A',W)$.
Then, it needs to hold that $p(5,c_2)=1$ and therefore $r'_1=r'_2=r'_3=0$ by 1-stability. Since the total payment for $c_1$ is $1$, it needs to hold that $b'_\ell\leq \frac{1}{3}$ for some voter $\ell \in \{1,2,3\}$. However, note that by budget-averaging it needs to hold that $b'_4=r'_4\geq \frac{\lvert W\rvert}{n}=\frac{2}{5}>\frac{1}{3}=b'_\ell$, since the only 1-stability constraint involving voter $4$ for candidate $c_4^*$ is tight only when $r'_4=1$. However, this implies that $\mathfrak{E}$ needs to violate monotonicity.
\end{proof}
One intuition for this incompatibility is that even the deletion of a single voter $i$'s approval can make a proportional committee strongly non-proportional for other voters--thereby restricting their residual by 1-stability--while voter $i$ may remain appropriately represented and thus, by budget-averaging, cannot receive a budget below the average payment.

\subsection{Payment Axioms}\label{sec:app-payments}
In this section, we discuss additional ideas for  axioms for payments, and describe challenges in formulating such requirements that are compatible with other desiderata.
As outlined in the main body, the purpose of payments is to describe how voters exercise their influence. Moreover, comparing the payments made by different voters toward the same candidate should reveal which voters played a crucial role in securing that candidate’s selection. A guiding intuition is that if, in the absence of a particular voter, there would no longer be a reason to select some candidate, then that voter should bear a larger share of the cost. Therefore, a natural first requirement would be the following: if there is a candidate $c\in W$ and voters $i,j\in V[c]$ such that there exists a candidate $c'\in C \setminus W$ with $V[c] \setminus \{i\} \subset V[c']$ but no such candidate exists for $j$, then $i$ should pay more for $c$ than $j$, i.e., $p(i,c)>p(j,c)$. The following example illustrates that, while locally appealing, this axiom may not be sensible to enforce from a more global perspective. 
 \begin{example}\label{appendix:ex:paymentPOimprov}
     Let $V=\{i_1,i_2,i_3,i_4, i_5\} \cup [100]$, $C=\{c_1,c_2, c_3,c_4\}$, and $W=\{c_1,c_2\}$. Candidate $c_1$ is approved by the voters in $\{i_1,i_2\}$, $c_2$ is approved by the voters in $\{i_3,i_4,i_5\}$, $c_3$ is approved by voters $\{i_1,i_2,i_3,i_4\}$, and $c_4$ is approved by $\{i_5\} \cup [100]$. Here, the axiom would require that $p(i_5,c_2)>p(i_4,c_2)$, as there is a Pareto improvement by exchanging $c_2$ for $c_3$ after deleting $i_5$'s approval for $c_2$, but no such candidate exists for $i_4$. However, this arguably does not reflect the structure of the instance, as the voters benefiting from this deviation are already well represented, and a deviation to $c_4$ benefiting the completely unrepresented voters $[100]$ is not captured by this requirement.
 \end{example}

Another intuition is that if one voter is generally ``easier'' to satisfy than another, then they should never pay more for a jointly approved candidate. One way to formalize this is to require that if one voter approves a superset of another voter (i.e., the voter is easier to satisfy),  then they should never pay more for a jointly approved candidate:  for any two voters $i,j\in V$ with $A_i\subseteq A_j$ and $c\in A_i\cap A_j \cap W$, it holds that
$
p(i,c) \geq p(j,c).
$
While the intuition behind this requirement is appealing, it can be very restrictive, and may rule out price systems that are arguably consistent with the structure of the instance.  
\begin{example}
    Let $V=\{i,j\} \cup [999] \cup \{x_1,x_2,x_3,x_4,y\}$, $C=(\{c^*_1,\dots,c^*_{201}\}\cup \{c, c_y, c_x\})$, where candidate $c^*_\ell\in C$ is approved by the voters $\{j\} \cup [999]$, candidate $c$ is approved by  the voters  $[990] \cup \{i,j\}$, candidate $c_y$ is approved by the voters $\{i,j,y\}$, and candidate $c_x$ is approved by $\{i,j,x_1,x_2,x_3,x_4\}$. The committee is $W=C \setminus\{c_x\}$.
    Consider the price system $\mathbf{ps}$ computed by \texttt{Continuous Phragm{\'e}n}:
    First, the voters $\{j\} \cup [999]$ pay for the candidates $c^*_1,\dots,c^*_{201}$, while $i$ and $y$ pay for $c_y$, and the voters $\{x_1,x_2,x_3,x_4\}$ accumulate residual.
    Note that before the candidates $c^*_1,\dots,c^*_{201}$ are fully funded, we have that $p(i,c_y)=\frac{1}{5}$ and $r_{x_1}=r_{x_2}=r_{x_3}=r_{x_4}=\frac{1}{5}$. Since $c_y \in C_i$, voter $i$ is blocked since otherwise there is a 1-stability violation for $c_x$. 
    Once the candidates $c^*_1,\dots,c^*_{201}$ are fully funded, voters $\{j\} \cup [990]$ pay for candidate $c$.  
    Hence, we have that $j$ has a non-zero payment for $c$, since $j$ is only blocked later when  $p(j,c)+r_{x_1}+r_{x_2}+r_{x_3}+r_{x_4}=1$, or when $c_y \in C_j$. We have that $p(i,c_y)=\frac{1}{5}>p(j,c)>0=p(i,c)$, which violates the requirement stated above. 

    However, when interpreted through a ``sequential'' intuition, the computed price system is arguably plausible: the reason why voter $i$ does not pay for $c$ is simply that they pay substantially more for another shared candidate $c_y$ (which $j$ does not pay for). Moreover, the payment of voter $j$ for $c$ is strictly smaller than $i$'s payment for this candidate $c_y$.   
\end{example}
We leave a further investigation of less demanding variants of this requirement, or of desirable explanation rules that satisfy it for future work.

\section{Proofs \& Additional Discussion for Section 4}\label{app:algos}
In this section, we prove the axiomatic guarantees of our explanation rules.
\subsection{\texttt{Approximate Priceability}}\sloppy
We begin by proving the axiomatic guarantees of \texttt{Approximate Priceability}.
\priceabilityGuarantees*
\begin{proof}
    To see that \texttt{Approximate Priceability} satisfies laminar-proportional-uniformity, perfect-coverage-uniformity, and perfect-symmetry-uniformity, note that by \Cref{thm:contPhragmenGuarantees}, the price system returned by \texttt{Continuous Phragm{\'e}n} is 1-stable (and thus in particular residual-stable), and budget-uniform for such committees. Hence, \texttt{Approximate Priceability} needs to return a budget-uniform price system, as it minimizes the sum of absolute pairwise differences between voters’ budgets among all residual-stable price systems.

    We proceed by showing that \texttt{Approximate Priceability} violates the following axioms.
    \paragraph{Equal Treatment of Equals}
    Let $V=\{1,2\}$ and $C=\{c\}$.
    Voters $1$ and $2$ both approve candidate $c$.
    The committee is $W=\{c\}$. Note that the following price system $\mathbf{ps}$ is budget-uniform and will thus be returned by \texttt{Approximate Priceability}: payments $p(1,c)=1$, $p(2,c)=0$, and residuals $r_1=0$ and $r_2=1$. This violates equal treatment of equals. 

    \paragraph{1-Stability and Single-Winner-Payment Responsiveness}
    Let $V=[2q]$ and $C=\{c_1,c_2\}$.
    The voters $V\setminus \{1\}$ approve candidate $c_1$, and only the voters in $[q]$ approve $c_2$. The committee is $W=\{c_2\}$. Observe that only $c_1$ is an approval winner. Note that the following price system $\mathbf{ps}$ is budget-uniform and will thus be returned by \texttt{Approximate Priceability}: payments $p(i,c_2)=\frac{1}{q}$ for $i\in [q]$, and residuals $r_i=0$ for $i\in [q]$ and $r_i=\frac{1}{q}$ for $i\in [q+1,2q]$. This violates both 1-stability and single-winner-payment-responsiveness. 

    \paragraph{$\Delta$-unproportional-responsiveness}
    To show that for any $\Delta \in \NN$, \texttt{Approximate Priceability} violates $\Delta$-unproportional-responsiveness, we fix $\Delta \in \NN$ and construct the following profile: 
    Let $V=[2]$ and $C=\{c^*_1,\dots, c^*_{\Delta+1}, c_1,\dots, c_{\Delta+1},c'_1,\dots, c'_{\Delta+1}\}$. All voters approve the candidates in $\{c^*_1,\dots, c^*_{\Delta+1}\}$, and the candidates $\{ c_1,\dots, c_{\Delta+1}\}$ and $\{c'_1,\dots, c'_{\Delta+1}\}$ are approved only by $1$ and $2$, respectively. Note that this is a laminar profile. 
    The committee is $W=\{c^*_1,\dots, c^*_{\Delta+1}, c_1,\dots, c_{\Delta+1}\}$. Note that $W$ is $\Delta$-laminar-unproportional.
    However, note that the following price system $\mathbf{ps}$ is budget-uniform and will thus be returned by \texttt{Approximate Priceability}: payments $p(2,c^*_\ell)=1$ for $c^*_\ell\in W$, $p(1,c_\ell)=1$ for $c_\ell\in W$, and residuals $r_i=0$ for $i\in V$. This violates $\Delta$-unproportional-responsiveness. 
    
\end{proof}

\subsection{\texttt{Continuous Phragm{\'e}n}}
In the following example, we discuss the simpler alternative of setting $C_i=A_i\cap W_{\mathrm{remaining}}$ for all $i\in V_{\mathrm{active}}$ in \texttt{Continuous Phragm{\'e}n}.
\begin{example}\label{ex:cont-phragmen-ci-alt}
    Let $V=[11]$, $C=\{c^*, c_1,c_2, c'_1, \dots, c'_{11}\}$,  where $c^*$ is approved by all voters,  $c_1,c_2$ are approved only by voter $1$,  and $c'_1, \dots, c'_{11}$ are approved by the voters $V\setminus \{1\}$.
    The committee is $W=\{c^*, c_1, c'_1, \dots, c'_{10}\}$, observe that $W$ is laminar proportional. 

    If we set $C_i=A_i\cap W_{\mathrm{remaining}}$ for all $i\in V_{\mathrm{active}}$, \texttt{Continuous Phragm{\'e}n} proceeds as follows.
    \begin{enumerate}
        \item Until $c^*$ is fully funded, in each iteration, $h_{c^*}$ is reduced by $(\frac{1}{2}+\frac{10}{11}) \cdot \varepsilon=\frac{31\varepsilon}{22}$. When $c^*$ is fully funded, we have $h_{c_1}=1-\frac{22}{31}\frac{1}{2}=\frac{20}{31}$ and $h_{c'_\ell}=1-\frac{22}{31}\frac{10}{11}=\frac{11}{31}$ for all $c'_\ell \in W$. 
        \item 
        After $c^*$ is fully funded, voter $1$ fully contributes to $c_1$, and the remaining ten voters split their earning equally over the candidates $c'_1, \dots, c'_{10}$, so $h_c$ is reduced by $\varepsilon$ for all $c\in W_{\mathrm{remaining}}$ in every iteration.
        Hence, when $h_{c'_\ell}=0$ for all $c'_\ell \in W$, we have that  $h_{c_1}=\frac{20}{31}-\frac{11}{31}=\frac{9}{31}$.
        \item After $h_{c'_\ell}=0$ for all $c'_\ell \in W$, the voters $V\setminus \{1\}$ accumulate residual until they are blocked when $r_i=\frac{1}{10}$ for all $i \in V\setminus \{1\}$. Note that we have that $h_{c_1}=\frac{9}{31}-\frac{1}{10}>0$ in this iteration, so voter $1$ continues earning and the resulting price system is not budget-uniform.
    \end{enumerate}
    As the resulting price system is not budget-uniform, laminar-proportional-uniformity is violated. The intuitive reason for this is that there are unjustified differences in the payments of the voters for the shared candidate $c^*$.
\end{example}
We proceed by proving the axiomatic guarantees for \texttt{Continuous Phragm{\'e}n}. 
Within some iteration in Line~\ref{line:outer-while}, we refer to the \emph{time step} (of this iteration) as the amount of cumulative $\varepsilon$-steps taken before it.

We begin by observing the following simple fact. 
\begin{observation}\label{app:obs:invertedAuto}
    For some profile $A \in \app$ and committee $W\subseteq C$, let $(\sigma,\pi)$ be an automorphism of $(A,W)$. Then, $(\sigma^{-1},\pi^{-1})$ is also an automorphism of $(A,W)$.
\end{observation}
\begin{proof}
    Clearly, $\sigma^{-1}$ and $\pi^{-1}$ exist and are bijections.  Next, we show that $(\sigma^{-1}, \pi^{-1})$ is a profile automorphism of $A$. Fix $i \in V$. Since $(\sigma, \pi)$ is a profile automorphism of $A$, we know that for $j = \sigma^{-1}(i)$, it holds that $\pi(A_j) = A_{\sigma(j)}$, so
    \[
        \pi(A_{\sigma^{-1}(i)}) = A_i.
    \]
    Applying $\pi^{-1}$ to both sides yields:
    \[
        A_{\sigma^{-1}(i)} = \pi^{-1}(A_i).
    \]
    Thus, $(\sigma^{-1}, \pi^{-1})$ is a profile automorphism of $A$.
    Moreover, since $\pi(W)=W$ and $\pi$ is a bijection, it follows that $\pi^{-1}(W)=W$.
\end{proof}

\contPhragmenGuarantees*
\begin{proof}
    Fix a profile $A \in \app$ and committee $W\subseteq C$, and let $\mathbf{ps}$ be the price system returned by \texttt{Continuous Phragm{\'e}n}.
    \paragraph{Symmetry}  
    We first prove the following invariants.
    \begin{claim}\label{claim:symmetry-invariant}
    For every automorphism $(\sigma, \pi)$ of $(A, W)$, at the beginning of each iteration of the outer loop (i.e., at the beginning of each iteration of the loop in Line~\ref{line:outer-while}), and at the end of each iteration of the outer loop (after Line~\ref{line:update-remaining}), the following invariants hold:
    \begin{enumerate}
        \item $i \in V_{\mathrm{active}} \iff \sigma(i) \in V_{\mathrm{active}}$ for all $i \in V$,
        \item $c \in W_{\mathrm{critical}} \iff \pi(c) \in W_{\mathrm{critical}}$ for all $c \in W$,
        \item $h_c = h_{\pi(c)}$ for all $c \in W$,
        \item $\pi(C_i) = C_{\sigma(i)}$ for all $i \in V_{\mathrm{active}}$,
        \item $p(i, c) = p(\sigma(i), \pi(c))$ for all $i \in V$ and $c \in W$,
        \item $p(i, i) = p(\sigma(i), \sigma(i))$ for all $i \in V$.
    \end{enumerate}
    \end{claim}

    \begin{claimproof}
    We proceed by induction over the iterations of the outer loop. At the beginning of the first iteration, the claim holds trivially from the initialization. For the inductive step, assume the claim holds at the beginning of an iteration, and consider any automorphism $(\sigma, \pi)$. We argue that the invariants are maintained within this iteration.

    \paragraph{\textsc{MoneyFlow}}
    First, consider the spending sets $(C_i)_{i \in V_{\mathrm{active}}}$ returned by \textsc{MoneyFlow} (Line~\ref{line:C-fund}). We now show that $c \in C_i \implies \pi(c)\in C_{\sigma(i)}$ for all $i \in V_{\mathrm{active}}$ and $c \in W$. Together with \Cref{app:obs:invertedAuto}, this implies $\pi(C_i) = C_{\sigma(i)}$ for all $i \in V_{\mathrm{active}}$ (hence, (4) is maintained).
    We show this by considering the execution of \textsc{MoneyFlow} and showing that at the beginning and end of each iteration therein, we have that $i\in V^* \iff \sigma(i) \in V^*$ and $c \in A_i \cap C' \iff \pi(c) \in A_{\sigma(i)} \cap C'$  for all $i \in V$ and $c \in W$. In the first iteration, this holds from the initialization of $V^*$: assume that $i\in V^*$, i.e., $i\in V_{\mathrm{active}}$ and there exists $c \in W \cap A_i$ with $h_{c}>0$. With invariants (1), (3), and (4), it follows that $\sigma(i)\in V_{\mathrm{active}}$, $\pi(c) \in W \cap A_{\sigma(i)}$, and $h_{\pi(c)}>0$, so $\sigma(i)\in V^*$ (and conversely with \Cref{app:obs:invertedAuto}). Next, consider any iteration and assume that $c \in A_i \cap C'$ for some $c\in W$. Note that 
    $$
        \frac{h_{\pi(c)}}{|V[\pi(c)]\cap V^*|}=\frac{h_{c}}{|\sigma(V[c]\cap V^*)|}=
        \frac{h_c}{|V[c]\cap V^*|},
    $$
    from invariant (3) and the fact that $\sigma(V[c])=V[\pi(c)]$ together with $i\in V^* \iff \sigma(i) \in V^*$.
    This implies that $\pi(c) \in C' \cap  A_{\sigma(i)}$.
    Again, the other direction follows with \Cref{app:obs:invertedAuto}. Moreover, this implies that $i\in V^* \iff \sigma(i) \in V^*$ is maintained (cf. Line~\ref{line:assign-ci}). 

    \paragraph{\textsc{Check Blocking}} 
    Let $V_{\mathrm{new}}$ be the set of voters returned by \textsc{Check Blocking}.
    We show that if $i\in V_{\mathrm{new}}$, then $\sigma(i)\in V_{\mathrm{new}}$.
    The other direction follows with \Cref{app:obs:invertedAuto}.
    If $i\in V_{\mathrm{new}}$, there is tight constraint for some  candidate $c\in C\setminus W$ where $i$ is added to $V^*$ in \textsc{Check Blocking}. 
    
    First, assume that this happens for $c\in C\setminus W$ with condition $\sum_{j\in V[c]}p(j,j)=1$. Hence, $i\in V[c]$, $A_i\cap W_{\mathrm{critical}}=\emptyset$ , and $C_i=\emptyset$. By the induction hypothesis, it holds for candidate $\pi(c)$ that $\sum_{j\in V[\pi(c)]}p(j,j)= \sum_{j\in \sigma(V[c])}p(j,j)=\sum_{j\in V[c]}p(\sigma(j),\sigma(j))= 1$. Moreover, $\sigma(i)\in V[\pi(c)]$, $A_{\sigma(i)}\cap W_{\mathrm{critical}}=\emptyset$ by (2), and $C_{\sigma(i)}=\emptyset$ by (4). Thus, $\sigma(i)$ is also added to $V^*$

    Second, assume that $i$ is added for $c\in C\setminus W$ with condition $$\sum_{j\in V[c] \setminus V[c']} p(j,j) + \sum_{j\in V[c']\cap  V[c] }  p(j,c') =1$$ for some $c'\in W$.
     By the invariants, it holds for candidates $\pi(c)$ and $\pi(c')$ that 
    \begin{align*}
        &\sum_{j\in V[\pi(c)] \setminus V[\pi(c')]} p(j,j) + \sum_{j\in V[\pi(c')]\cap V[\pi(c)]} p(j,\pi(c')) \\
        &= \sum_{j\in \sigma(V[c] \setminus V[c'])} p(j,j) + \sum_{j\in \sigma(V[c']\cap V[c])} p(j,\pi(c')) \\
        &= \sum_{j\in V[c] \setminus V[c']} p(\sigma(j),\sigma(j)) + \sum_{j\in V[c']\cap V[c]} p(\sigma(j),\pi(c'))=1.
    \end{align*}
    If $i \in V[c] \setminus V[c']$, it needs to hold that $A_i\cap W_{\mathrm{critical}}=\emptyset$ and $C_i=\emptyset$ since $i$ is added. 
    Then, $\sigma(i) \in V[\pi(c)] \setminus V[\pi(c')]$, $A_{\sigma(i)} \cap W_{\mathrm{critical}} = \emptyset$ by (2), and $C_{\sigma(i)} = \emptyset$ by (4). Thus, $\sigma(i)$ is also added to $V^*$.
    If $i \in V[c] \cap V[c']$, it needs to hold that $A_i\cap W_{\mathrm{critical}}=\emptyset$ and $c' \in C_i$ since $i$ is added. Then, $\sigma(i) \in V[\pi(c)] \cap V[\pi(c')]$, $A_{\sigma(i)} \cap W_{\mathrm{critical}} = \emptyset$ by (2), and $\pi(c') \in C_{\sigma(i)}$ by (4). Thus, $\sigma(i)$ is also added to $V^*$. Hence, (1) is maintained.

    \paragraph{Computation of $W_{\mathrm{critical}}$}
    If $c \in W_{\mathrm{unsupp}}$, then $c\in W$, $h_c>0$, and $V[c]\cap V_{\mathrm{active}}=\emptyset$. This implies that $\pi(c)\in W$, $h_{\pi(c)}=h_c>0$, and $V[\pi(c)]\cap V_{\mathrm{active}}=\sigma(V[c])\cap V_{\mathrm{active}}=\emptyset$ by (1), so $\pi(c)\in W_{\mathrm{unsupp}}$ (and conversely with  \Cref{app:obs:invertedAuto}). From this, it follows that Line~\ref{line:crit} maintains (1) and (2). 
    
    \paragraph{Spending Phase}
    Hence, all invariants hold at the beginning of the spending phase. This implies that the invariants hold after the spending phase: it holds that $i \in V_{\mathrm{active}} \iff \sigma(i) \in V_{\mathrm{active}}$ for every $i\in V$, and $|C_i|=|C_{\sigma(i)}|$ and $c \in C_i \iff \pi(c) \in C_{\sigma(i)}$ for every $i\in V_{\mathrm{active}}$ and $c\in W$. Thus, for all $i\in V$ and $c \in W$, either both $i$ and $\sigma(i)$ are inactive, both contribute $\varepsilon$ to their residual, or they contribute the same to $c$ and $\pi(c)$, respectively.

    \paragraph{Pruning Phase}
    We argue that $t_i=t_{\sigma(i)}$ for all $i\in V$. To see that $t_{\sigma(i)}\leq t_i$ for all $i\in V$, fix $i\in V$ and  assume there exist $c \in C\setminus W$ and $c'\in W$ such that $i\in  V[c] \setminus V[c']$ and $\delta\coloneqq (\sum_{j\in V[c] \setminus V[c']} p(j,j) + \sum_{j\in V[c']\cap  V[c] }  p(j,c') )-1>0$.
    Then, it holds for $\pi(c)$ and $\pi(c')$ that 
    \begin{align*}
        \quad&\sum_{j\in V[\pi(c)] \setminus V[\pi(c')]} p(j,j) + \sum_{j\in V[\pi(c')]\cap  V[\pi(c)] }  p(j,\pi(c')) - 1 \\
        &=  \sum_{j\in \sigma(V[c] \setminus V[c'])} p(j,j) + \sum_{j\in \sigma(V[c']\cap  V[c]) }  p(j,\pi(c')) - 1\\
        &= \sum_{j\in V[c] \setminus V[c']} p(\sigma(j),\sigma(j)) + \sum_{j\in V[c']\cap  V[c] }  p(\sigma(j),\pi(c')) - 1=\delta, 
    \end{align*}
    and $\sigma(i)\in  V[\pi(c)] \setminus V[\pi(c')]$.
    The other direction holds by \Cref{app:obs:invertedAuto}. Hence, invariant (6) is maintained.

    \paragraph{Update of $W_{\mathrm{remaining}}$ and $W_{\mathrm{critical}}$}
    Clearly, this steps preserves the invariants as $h_c=h_{\pi(c)}$ for all $c\in W$.
    \end{claimproof}
     
    It remains to show that the residual phase preserves symmetry. 
    \begin{claim}\label{claim:symmetry-res}
        If  $W_{\mathrm{remaining}}=\emptyset$ and \begin{enumerate*}[label=(\alph*)]
        \item $p(i, c) = p(\sigma(i), \pi(c))$ for all $i \in V$ and $c \in W$, and
        \item $p(i, i) = p(\sigma(i), \sigma(i))$ for all $i \in V$
    \end{enumerate*} at the beginning of  \textsc{Distribute Residual}, then the following invariants hold at the beginning and end of each iteration in the loop in \textsc{Distribute Residual}:
    \begin{enumerate}
        \item $i \in V_{\mathrm{active}} \iff \sigma(i) \in V_{\mathrm{active}}$ for all $i \in V$,
        \item $p(i, c) = p(\sigma(i), \pi(c))$ for all $i \in V$ and $c \in W$,
        \item $p(i, i) = p(\sigma(i), \sigma(i))$ for all $i \in V$.
    \end{enumerate}
    \end{claim}
    \begin{claimproof}
        Proof by induction. At the beginning of the first iteration, the claim follows directly from the initialization and the fact that (a) and (b) hold at the beginning of  \textsc{Distribute Residual}.
        Let $V_{\mathrm{new}}$ be the set of voters returned by \textsc{Check Blocking}.
        The argument to show that $i\in V_{\mathrm{new}}$ implies $\sigma(i)\in V_{\mathrm{new}}$ for all $i\in V$ is analogous to \Cref{claim:symmetry-invariant}, using that $C_i=\emptyset = C_{\sigma(i)}$.
        Since $i \in V_{\mathrm{active}} \iff \sigma(i) \in V_{\mathrm{active}}$ and $b_i=b_{\sigma(i)}$ for all $i \in V$ in Line~\ref{line:distr-res}, it follows that $p(i,i)$ is increased by $\varepsilon$ if and only if $p(\sigma(i),\sigma(i))$ is increased by the same amount.
    \end{claimproof}

    This completes the proof for symmetry. 

    \paragraph{Perfect-Symmetry-Uniformity}
    Assume that all candidates in $C$ are pairwise isomorphic and $W$ contains the same number of candidates from each type. Thus, for every pair of candidates $c,c'\in C$, there exists a profile automorphism $(\sigma, \pi)$ of $A$ with $\pi(c) = c'$. Note that we can construct an automorphism $(\sigma', \pi')$ for $(A,W)$ from $(\sigma, \pi)$ (i.e., ensure that $\pi'(W)=W$) by permuting candidates within the types, since $W$ contains the same number of candidates from each type. Hence, we can apply \Cref{claim:symmetry-invariant}. It follows that $h_c=h_{c'}$ for all $c,c'\in W$ at the end of each iteration. 
    Next, we argue that, in every iteration, no voter is blocked by \textsc{Check Blocking}:
    First, observe that since every voter approves some candidate $c\in C$, they also approve a candidate $c'\in W$ of the same type. Since it needs to hold that $h_{c'}>0$ (otherwise, $h_{c^*}=0$ for all $c^*\in W$ and the loop in Line~\ref{line:outer-while} terminates), we have that $C_j\neq\emptyset$ for all $j\in V$ in every iteration. Since therefore $r_j=0$ for all $j\in V$ and $\sum_{j'\in V[c^*]} p(j',c^*) < 1$ for all $c^*\in W$, it follows that no voter is blocked by \textsc{Check Blocking}, as no constraint can be tight.  
    From this, it follows that $\mathbf{ps}$ is budget-uniform, since $b_i=b_j$ for all $i,j\in V$ (as both $i$ and $j$ are active in all iterations) before \textsc{Distribute Residual} in Line~\ref{line:dis}, so no additional residual is distributed.
            
    \paragraph{1-Stability}
    We first consider the price system computed by the algorithm before entering the residual phase in Line~\ref{line:dis}. Clearly, since residuals are pruned as in Line~\ref{line:pruning-end} and $\sum_{i\in V[c']\cap  V[c] }  p(i,c') \leq 1$, there cannot exist any candidates $c \in C\setminus W$ and $c'\in W$ with $\sum_{i\in V[c] \setminus V[c']} p(i,i) + \sum_{i\in V[c']\cap  V[c] }  p(i,c') > 1$.
    Next, suppose there is a candidate  $c \in C\setminus W$ with $\sum_{i\in V[c]} p(i,i)>1$.
    Then, as residuals are increased continuously, at some point in the algorithm, it needs to hold that $\sum_{i\in V[c]} p(i,i)=1$ and $C_j=\emptyset$ for some $j\in V_{\mathrm{active}}\cap V[c]$ in Line~\ref{line:spending-start}. However, this cannot be since $j$ is blocked by \textsc{Check Blocking} in Line~\ref{line:V-blocked} in the same iteration, and $C_j=\emptyset$ implies that $A_j\cap W_{\mathrm{remaining}}=\emptyset$ so $j$ is not unblocked in Line~\ref{line:crit}. Therefore, it needs to hold that $j\notin V_{\mathrm{active}}$ in Line~\ref{line:spending-start}.

    To see that the residual phase preserves 1-stability, note that voters are removed from $ V_{\mathrm{active}}$ in Line~\ref{line:block-res} if increasing the residuals violates 1-stability, and voters are never added to $ V_{\mathrm{active}}$ once removed.

    \paragraph{Single-Winner-Payment-Responsiveness}
    Suppose that $W=\{c\}$ for some candidate $c\in C$ and $c$ is not an approval winner. Let $i,j\in V[c]$ where $i$ does not approve an approval winner and $j$ approves some approval winner $c^*\in C \setminus W$. We need to show that $p(i,c)>p(j,c)$. 
    Note that up to the time step $t=\frac{1}{|V[c^*]|}$,  no voters will be blocked, since for all $c' \in C \setminus W$, $|V[c']|\leq|V[c^*]|$  and $r_{j'}<\frac{1}{|V[c^*]|}$ and $p(j',c'')<\frac{1}{|V[c^*]|}$ for all $j'\in V$ and $c''\in W$. 
    Hence, up to $t$,  it holds that $V_{\mathrm{active}}=V$. Moreover, for all $j'\in V$ we have $C_{j'}=\{c\}$ if $c\in A_{j'}$ and $C_{j'}=\emptyset$ otherwise since $c$ is the only selected candidate. Therefore $p(j',c)=t$ for all $j'\in V[c]$ and $r_{j'}=t$ for all $j'\in V\setminus V[c]$ at the beginning of the iteration with time step $t$.
    We consider the iteration for time step $t$. Note that it needs to hold $h_c>0$, since $c$ is not an approval winner, so $|V[c]|<|V[c^*]|$. Hence $C_{i}=C_{j}=\{c\}$. 
    Let $V^*$ be the set of voters returned by \textsc{Check Blocking}.
    Since it holds for candidate  $c^*\in A_j \setminus W$ that 
    $$
    \sum_{j'\in V[c^*] \setminus V[c]} p(j',j') + \sum_{j'\in V[c^*]\cap  V[c] }  p(j',c)=1,
    $$
    we have $j\in V^*$. Note that $i\notin V^*$, as $|V[c']|<|V[c^*]|$ for all $c'\in A_i \setminus W$. Thus, $W_{\mathrm{critical}}=\emptyset$.
    Hence, $i\in V_{\mathrm{active}}$, $j\notin V_{\mathrm{active}}$, and $C_{i}=\{c\}$ at the beginning of the spending phase in Line~\ref{line:spending-start}, 
    Therefore, after this iteration we have
    \begin{equation}\label{eq:p-c-ineq}
        p(i,c)\geq p(j,c) + \varepsilon > p(j,c).
    \end{equation}
    Note that in all later iterations  until $c$ is fully funded, we have that $j\in V_{\mathrm{active}}$ implies $i\in V_{\mathrm{active}}$: voter $i$ is reactivated only if $c\in W_{\mathrm{critical}}$, which implies that also $i$ is reactivated and both $i,j\in V_{\mathrm{active}}$ from this iteration until $c$ is fully funded.
    Hence, \Cref{eq:p-c-ineq} holds when the algorithm terminates.

    \paragraph{Budget-Averaging}
     Consider any voter $i\in V$ with $b_i<\frac{|W|}{n}$. Clearly, it needs to hold that $b_i<\max_{j\in V} b_j$. Hence, in \textsc{Distribute Residual}, voter $i$ needs to be blocked (as otherwise they would continue to accumulate residual until $b_i=\max_{j\in V} b_j$), which implies that there is some candidate $c \in C \setminus W$ for which increasing $i$'s residual creates a 1-stability violation.

     \paragraph{Laminar-Proportional-Uniformity and 1-Unproportional-Responsiveness}
     We first prove the following invariant, which together with symmetry implies that $\mathbf{ps}$ satisfies \lnice{} when $A$ is laminar.
     \begin{claim}\label{claim:laminar}
         When $A$ is a laminar profile, at the beginning of each iteration of the loop in Line~\ref{line:outer-while}, and before Line~\ref{line:update-remaining}, it holds for all $c\in W$ and (not necessarily distinct) $i,j\in V[c]$ that:
        \begin{enumerate}
            \item If $c\in W_{\mathrm{remaining}}$, then $i,j\in V_{\mathrm{active}}$, $c'\in C_i \iff c'\in C_j$ for all $c'\in W$, $h_{c'}=1$ for all $c'\in W$ with $V[c']\subset V[c]$,  and $r_{j'}=0$ for all $j'\in V[c]$.
            \item It holds that $p(i,c)=p(j,c)$.
        \end{enumerate}
     \end{claim}
     \begin{claimproof}
         We proceed by induction over the iterations of the outer loop. At the beginning of the first iteration, the claim holds trivially from the initialization. For the inductive step, assume the claim holds at the beginning of an iteration. Assume that $c\in W_{\mathrm{remaining}}$, otherwise we are done. 
         First, we consider the computation of the spending sets $(C_\ell)_{\ell \in V_{\mathrm{active}}}$ in \textsc{MoneyFlow} (Line~\ref{line:C-fund}).
         
         Using an inductive argument, we show that 
         \begin{enumerate*}[label=(\alph*)]
        \item $c'\notin C_i$ for all $c'\in W$ with $V[c']\subset V[c]$, 
        \item $V[c] \subseteq V^*$ or $V[c] \cap V^*=\emptyset$, and
        \item $c' \in C_i$ implies that $c' \in C_j$ for all $c'\in W$:
        \end{enumerate*}
         For the induction step, consider an iteration in \textsc{MoneyFlow} and assume that $V[c] \subseteq V^*$, otherwise we are done.
         Assume for the sake of contradiction that some $c'\in W$ with $V[c']\subset V[c]$ is added to $C_i$ in Line~\ref{line:assign-ci} in this iteration.
         This implies that in this iteration
         $$
             \frac{h_{c'}}{|V[c']\cap V^*|}\leq\frac{h_c}{|V[c]\cap V^*|},
         $$
         as otherwise $c'$ would not be added to $C_j$ in Line~\ref{line:assign-ci}. However, this directly leads to a contradiction, since $h_{c'}=1\geq h_c$ and $|V[c']\cap V^*|< |V[c]\cap V^*|$ as $V[c']\subset V[c] \subseteq V^*$.
         
         Next, assume there is a candidate $c'\in W$ such that $c'$ is added to $C_i$ but not to $C_j$. It must hold that $c'\notin A_j$.  Since $A$ is laminar, $c'\notin A_j$ implies that $V[c']\subset V[c]$, which leads to a contradiction as argued above.
         Therefore, both $i$ and $j$ are assigned the same candidates in Line~\ref{line:assign-ci}, so  (b) and (c) are maintained.

         Moreover, note that $C_\ell\neq \emptyset$ for all $\ell \in V[c]$ since $c\in W_{\mathrm{remaining}}$.

         We consider the set of voters $V^*$ returned by \textsc{Check Blocking} in Line~\ref{line:V-blocked} next. We will argue that if $i\in V^*$, then it needs to be that $j\in V^*$. Suppose that $i$ is added in \textsc{Check Blocking} for some candidate $c'\in A_i\setminus W$. Since $C_i\neq\emptyset$, voter $i$ can only be added if there is a candidate $c^* \in C_i$ with
        \begin{equation}\label{eq:sum-cprime}
            \sum_{\ell\in V[c'] \setminus V[c^*]} p(\ell,\ell) + \sum_{\ell\in V[c']\cap  V[c^*] }  p(\ell,c^*) =1.
        \end{equation}
        As shown above, we have that  $c^* \in C_j$. It remains to show that $c'\in A_j$. Note that if $c'\notin A_j$, then $V[c']\subset V[c]$ as $A$ is laminar. This however contradicts \Cref{eq:sum-cprime}, as by (1)
        $$
        \sum_{\ell\in V[c'] \setminus V[c^*]} p(\ell,\ell) \leq \sum_{\ell\in V[c]} p(\ell,\ell) =0, 
        $$
        and as it needs to be that $h_{c^*}>0$ as $c^* \in C_i$
        $$
         \sum_{\ell\in V[c']\cap  V[c^*] }  p(\ell,c^*) \leq\sum_{\ell\in  V[c^*] }  p(\ell,c^*) <1.
        $$
        
        Since we have shown above that it holds for arbitrary $i,j\in V[c]$ that if $i$ is added to $V_{\mathrm{blocked}}$, then also $j$ is added, it directly follows that $c \in W_{\mathrm{unsupp}}$ if $i\in V_{\mathrm{blocked}}$. 
        Hence, we have that $i,j\in V_{\mathrm{active}}$ after Line~\ref{line:crit}. 

         Therefore, we have that $c'\notin C_\ell$ for all $\ell \in V$ and $c'\in W$ with $V[c']\subset V[c]$, and $C_\ell\neq \emptyset$ for all $\ell \in V[c]$. Moreover, all invariants hold at the beginning of Line~\ref{line:spending-start}. Thus, the invariants are maintained by the spending phase.
     \end{claimproof}
    \sloppy 
     With this claim, it follows from \Cref{prop:lnice-residual-unpropresp} that \texttt{Continuous Phragm{\'e}n} satisfies 1-unproportional-responsiveness.
     If $W$ is laminar proportional, we need to show that $\mathbf{ps}$ is budget-uniform.
     With the property from \Cref{claim:laminar} that the cost of a candidate $c\in W$ is split equally among the voters in $V[c]$, it follows that $p_i=p_j$ for all $i,j\in V$ using a simple structural induction on the laminar profile:
     In the base case where the profile is unanimous, all payments are clearly equal. The induction step for adding a unanimous candidate is clear, since all voters pay the same amount for the unanimous candidate. When taking the union of two disjoint profiles with voters $V_1,V_2$ and candidates $C_1,C_2$, let $p_1^*, p^*_2$ be the payment of voters in $V_1$ and $V_2$ for $C_1$ and $C_2$, respectively. Since $W$ is laminar proportional, we have that 
     \begin{align*}
         \frac{|V_1|}{|W\cap C_1|}=\frac{|V_2|}{| W\cap C_2|} \iff \frac{|V_1|}{|V_1|p_1^*}=\frac{|V_2|}{|V_2|p_2^*} \iff p_1^* = p_2^*.
     \end{align*}
     Note that it now suffices to show that $r_i=0$ for all $i\in V$ before \textsc{Distribute Residual} in Line~\ref{line:dis}, since \textsc{Distribute Residual} will then not allocate additional residual (as all budgets are equal) and $\mathbf{ps}$ will hence be budget-uniform. 
     Consider the first iteration where after Line~\ref{line:update-remaining} we have $A_i\cap W_{\mathrm{remaining}}=\emptyset$ for some $i\in V$. Note that by \Cref{claim:laminar}~(1), it follows that $V_{\mathrm{active}}=V$ when entering the spending phase in Line~\ref{line:spending-start} in all iterations up to (including) this iteration, as for every voter $j\in V$ there is some $c\in A_j\cap W_{\mathrm{remaining}}$. Moreover, we have that $r_j=0$ for all  $j\in V$, since \Cref{claim:laminar}~(1) holds in this iteration before Line~\ref{line:update-remaining}.
     Let $p^*=p_i$ be the total payment of $i$. Note that since $r_j=0$ for all  $j\in V$ and $V_{\mathrm{active}}=V$ in every spending phase, it follows that $p_j=p^*$ for all $j\in V$. Since we argued before that the total payment is equal for all voters, it follows that $h_c=0$ for all $c\in W$, so $W_{\mathrm{remaining}}=\emptyset$ and the condition for the loop in Line~\ref{line:outer-while} is no longer met. Hence, we have $r_j=0$ for all $j\in V$ before \textsc{Distribute Residual} in Line~\ref{line:dis}.

     \paragraph{Perfect-Coverage-Uniformity}
     Let $W$ be a committee that provides perfect coverage. We show that $\mathbf{ps}$ is budget-uniform by proving the following stronger claim. 
     \begin{claim}
         When $W$ provides perfect coverage, at the beginning and end of each iteration of the loop in Line~\ref{line:outer-while}, it holds that
         \begin{enumerate}
             \item $p(i,c)=p(j,c)$ for every $c\in W$ and $i,j \in V[c]$, 
             \item $r_i=0$ for $i\in V$ with $A_i\cap W_{\mathrm{remaining}}\neq \emptyset$,
             \item $V_{\mathrm{active}}=V$, and
             \item $b_i=b_j$ for all $i,j\in V$.
         \end{enumerate}
     \end{claim}
     \begin{claimproof}
          We proceed by induction over the iterations of the outer loop. At the beginning of the first iteration, the claim holds trivially from the initialization. For the inductive step, assume the claim holds at the beginning of an iteration.
          Fix $i\in V$. As $W$ provides perfect coverage, we have that $A_i\cap W = \{c\}$ for some $c\in W$. Hence,  for the spending sets $(C_\ell)_{\ell \in V_{\mathrm{active}}}$ in \textsc{MoneyFlow} (Line~\ref{line:C-fund}), it holds that $C_i= \{c\}$ if $c\in W_{\mathrm{remaining}}$ and  $C_i=\emptyset$ otherwise. In particular, this implies that $C_j=C_{j'}$ for all $c'\in W$ and $j,j'\in V[c']$. 
          
          Next, we show that $i\notin V_{\mathrm{blocked}}$ after Line~\ref{line:V-blocked}. 
          Let $c^*\in W$ be a selected candidate with minimum size $|V[c^*]|$. It needs to hold that $h_{c^*}>0$, since otherwise $h_{c'}=0$ for all $c'\in W$ and $W_{\mathrm{remaining}}=\emptyset$.
          By (1), it holds that $p(j,c^*)=p(j', c^*)$ for all $j,j'\in V[c^*]$.
          Consider voter $j\in V[c^*]$.
          Since $h_{c^*}>0$, it follows that 
          $$
          p(j,c^*)<\frac{1}{|V[c^*]|},
          $$
          and using (2) we get
          $$
          b_j= p(j,c^*)<\frac{1}{|V[c^*]|}.
          $$
          Since by (4) all budgets are equal (all voters are active in all iterations), we get that for all $j'\in V$
          $$
          r_{j'}\leq b_{j'}=b_j<\frac{1}{|V[c^*]|}.
          $$
        
          Consider any $c'\in A_i\setminus W$ and note that $|V[c']|\leq|V[c^*]|$ since $W$ provides perfect coverage. We thus have
          $$
            \sum_{\ell\in V[c']} p(\ell,\ell) < \sum_{\ell\in V[c']}\frac{1}{|V[c^*]|} = \frac{|V[c']|}{|V[c^*]|} \leq 1,
          $$
          and 
          $$
              \sum_{\ell\in V[c'] \setminus V[c]} p(\ell,\ell) + \sum_{\ell\in V[c']\cap  V[c] }  p(\ell,c) < \sum_{\ell\in V[c']}\frac{1}{|V[c^*]|} = \frac{|V[c']|}{|V[c^*]|} \leq 1,
          $$
          so voter $i$ is not blocked.
          
          Since $i\in V_{\mathrm{active}}$, $C_i= \{c\}$ if $c\in W_{\mathrm{remaining}}$ and  $C_i=\emptyset$ otherwise, and $C_{j'}=C_{j''}$ for all $c'\in W$ and $j',j''\in V[c']$, the spending phase maintains the invariants. Note that $t_{j'}=1$ for all $j'\in V$, since no 1-stability constraint can be violated as argued above. 
     \end{claimproof}
     
     Note that it follows with the above claim that $\mathbf{ps}$ is budget-uniform, since $b_i=b_j$ for all $i,j\in V$ before \textsc{Distribute Residual} in Line~\ref{line:dis}, so no additional residual is distributed.
\end{proof}

\paragraph{Implementation of \Cref{alg:cont-phragmen}}
To implement \Cref{alg:cont-phragmen} (cf. code available on \href{https://github.com/luca-kreisel/Explanation-Systems}{GitHub}), before the spending phase in each iteration of the outer loop, we compute the minimum time step $\delta_t$ until
\begin{enumerate*}[label=(\alph*)]
    \item a candidate $c \in W_{\mathrm{remaining}}$ is fully paid for, or
    \item continuing funding according to the current spending sets $(C_i)_{i\in V_{\mathrm{active}}}$ creates a 1-stability violation (involving at least one voter not approving a candidate in  $W_{\mathrm{critical}}$), i.e., \textsc{Check Blocking} returns a non-empty set of voters that need to be blocked. 
\end{enumerate*}
Then, we apply the spending phase for this computed time step $\delta_t$. 
Note that there can be at most $\bigO(n\cdot \lvert W\rvert)$ such steps, as in each step we need to fully pay for a candidate $c\in W$, or add a new voter to $V_{\mathrm{blocked}}$ (note that voters are removed from $V_{\mathrm{blocked}}$ only if they approve a candidate $c \in W_{\mathrm{critical}}$, and can then only be re-added once $h_c=0$).
Moreover, since we add candidates to $W_{\mathrm{critical}}$ and unblock their supporters when they do not have any active supporters, all selected candidates will be fully paid for. 
Furthermore, note that using this argument, it follows that there can in total be at most $\bigO(n\cdot \lvert W\rvert)$ iterations of the convergence-loop in the planning phase, amortized over all iterations of the outer loop.

Similarly, we implement the residual phase by computing the minimum time step $\delta_t$ until
\begin{enumerate*}[label=(\alph*)]
    \item allocating additional residual to a voter $i\in V_{\mathrm{active}}$ violates 1-stability, so $i$ is added to $V_{\mathrm{blocked}}$, or 
    \item we need to add a voter to the set $V^*=\{j \in V_{\mathrm{active}} \mid b_j = \min_{i\in V_{\mathrm{active}}} b_i \}$ (or all voters in $V^*$ reach the maximum budget).
\end{enumerate*}
Then, we add residual for all voters in $V^*$ for this computed time step $\delta_t$. Note that voters are never removed from $V_{\mathrm{blocked}}$ in this phase, and each voter can only be added once to $V^*$ , as they are only removed from $V^*$ when they are blocked and then remain blocked until termination. Hence we have at most $\bigO(n)$ such steps. 

All other operations within the loops can be executed in $\bigO(n^2\cdot m^2\cdot \lvert W\rvert)$.
Thus, the total running time is in $\bigO(n^3\cdot m^2\cdot \lvert W\rvert^2)$.

\subsection{\texttt{Equal Split}}
Finally, we consider \texttt{Equal Split}.
\eqSplitGuarantees*
\begin{proof}
    Fix a profile $A \in \app$ and committee $W\subseteq C$, and let $\mathbf{ps}$ be the price system returned by \texttt{Equal Split}.
    
    Recall that residuals are distributed as in \textsc{Distribute Residual} in \Cref{alg:cont-phragmen}.
    As 1-stability is trivially satisfied when $r_i=0$ for all $i\in V$, it remains to show that \textsc{Distribute Residual} preserves 1-stability, which follows analogously as in the proof of \Cref{thm:contPhragmenGuarantees}.
    Budget-averaging also depends only on \textsc{Distribute Residual} and hence follows analogously as in the proof of  \Cref{thm:contPhragmenGuarantees}.
    For symmetry, let $(\sigma, \pi)$ be an automorphism of $(A, W)$. Clearly, we have that $p(i, c) = p(\sigma(i), \pi(c))$ for all $i \in V$ and $c \in W$. 
     It remains to show that $p(i, i) = p(\sigma(i), \sigma(i))$ for all $i \in V$, which follows with \Cref{claim:symmetry-res} from the proof of \Cref{thm:contPhragmenGuarantees}.

    Note that if $A$ is a laminar profile, then $\mathbf{ps}$ satisfies \lnice{}, as all supporters $i\in V[c]$ have the same payment for every candidates $c\in W$, and all residuals are equal for voters $i,j\in V$ with $A_i=A_j$ by symmetry.
    Hence, with residual-stability, it follows from \Cref{prop:lnice-residual-unpropresp} that \texttt{Equal Split} satisfies $1$-unproportional-responsiveness. Moreover, if $W$ is laminar proportional, then $p_i=p_j$ for all $i,j\in V$. Hence, no residual is distributed in \textsc{Distribute Residual} and $\mathbf{ps}$ is budget-uniform, so \texttt{Equal Split} satisfies laminar-proportional-uniformity. 

    It remains to consider perfect-coverage-uniformity.
    We need to consider the computation of the residuals in \textsc{Distribute Residual}.
    Using an inductive argument, we show that $V_{\mathrm{active}}=V$ in \textsc{Distribute Residual} in every iteration, so \textsc{Distribute Residual} terminates when all voters have equal budget:
    Since $r_i=0$ for all $i\in V$ at the beginning of \textsc{Distribute Residual}, it holds for all $i^*\in V$ with $b_{i^*}=\max_{j\in V} b_j$ that $r_{i^*}=0$ throughout, since \textsc{Distribute Residual} terminates when all voters have reached the maximum budget. Let $c^*\in W$ be a selected candidate with minimum size $|V[c^*]|$. We have that $\max_{j\in V} b_j=\max_{j\in V} p_j=\frac{1}{|V[c^*]|}$ in every iteration.
    Suppose that a voter $i\in V_{\mathrm{active}}$ is blocked in some iteration. By the condition in the loop in \textsc{Distribute Residual}, we have that $\frac{1}{|V[c^*]|}=\max_{j\in V} b_j>\min_{j\in V_{\mathrm{active}}} b_j =\min_{j\in V} b_j$. Hence, $r_{j^*}<\frac{1}{|V[c^*]|}$ for all $j^*\in V$ with $b_{j^*}= \min_{j\in V} b_j$. Since $r_{j'}=0$ for all voters $j'\in V$ with $b_{j'}>\min_{j\in V_{\mathrm{active}}} b_j =\min_{j\in V} b_j$, we have that $r_{j'}<\frac{1}{|V[c^*]|}$ for all $j'\in V$.

    Voter $i\in V$ needs to be added to $V^*$ in \textsc{Check Blocking} for some  candidate $c \in A_i\setminus W$.
    Note that  $|V[c^*]|\geq |V[c]|$ as $W$ provides perfect coverage. 
    It cannot be that $i$ is added for condition $\sum_{j\in V[c]} p(j,j)=1$, as
    $$
        \sum_{j\in V[c]} p(j,j) < \sum_{j\in V[c]}\frac{1}{|V[c^*]|} = \frac{|V[c]|}{|V[c^*]|} \leq 1.
    $$
    Hence, it needs to hold that there exists a candidate $c'\in W$ such that 
    \begin{equation*}
        \sum_{j\in V[c]\setminus V[c']} p(j,j) +  \sum_{j\in V[c]\cap V[c']} p(j,c') =1.
    \end{equation*}
    Since $i$ is added to $V^*$ and $C_i=\emptyset$, it needs to hold that $i\in V[c]\setminus V[c']$. However, this leads to a contradiction: 
    Since $i\in V[c]\setminus V[c']$, we have that 
\begin{align*}
    \sum_{j\in V[c]\setminus V[c']} p(j,j) +  \sum_{j\in V[c]\cap V[c']} p(j,c') &= p(i,i) +  \sum_{j\in V[c]\setminus (V[c']\cup \{i\})} p(j,j) +  \sum_{j\in V[c]\cap V[c']} p(j,c') \\
    &\leq p(i,i) +  \frac{|V[c]|-1}{|V[c^*]|} < \frac{1}{|V[c^*]|}  +  \frac{|V[c]|-1}{|V[c^*]|} \leq 1.
\end{align*}
\end{proof}
\paragraph{Weighted Split}
A possible refinement of \texttt{Equal Split} is to change how each candidate’s cost is shared among its supporters. One approach for this, which we term \texttt{Weighted Split}, is the following: For each voter $i$, we define their total obligation as 
$ o_i = \sum_{c \in A_i \cap W} \frac{1}{|V[c]|}$.
Then, for each selected candidate $c \in W$, we distribute $c$'s cost among its supporters $i \in V[c]$ inversely proportional to their total obligation, i.e.,
$
  p(i,c) \propto \frac{1}{o_i}
  \text{for all } c \in W \text{ and } i \in V[c]$.
We leave an investigation of this approach for future work.

\subsection{Alternative Optimization-Based Approaches}\label{app:subsec:opt}
In the following, we discuss alternative, optimization-based approaches for computing price systems. Convex programming is particularly appealing because a strictly convex objective combined with a convex feasible set ensures that variables with symmetric constraints and roles in the objective function have equal values in an optimal solution. The main challenge lies in identifying an appropriate objective function.

We explore the following objective function:
\begin{equation*}
\text{Maximize} \sum_{i\in V} log(r_i+\varepsilon) + \sum_{c \in W, i\in V[c]} log(p(i,c)+\varepsilon),\tag{\texttt{Convex Program}}
\end{equation*}
for some small $\varepsilon>0$ (otherwise, the objective is undefined when the residual of some voter needs to be zero). The objective function can be expressed using exponential cone constraints, supported by standard solvers.  

We optimize the objective function subject to linear feasibility constraints encoding 1-stability.  
However, note that for $i\in V$ with $A_i\subseteq W$,  the residual for voter $i$ is unbounded. Hence, we need to introduce an additional constraint that limits $r_i\leq |W|$ for all $i\in V$.
We then ``prune'' excess residual as follows: let $b^* = \max(\max_{i\in V : r_i \leq 1} b_i, \max_{i\in V} p_i)$ if there is a ``bounded'' voter $i\in V$ with $ r_i \leq 1$, and $b^* = \max_{i\in V} p_i$ otherwise.
Then, for each $i\in V$ with $b_i > b^*$, we reduce $r_i$ by $b_i - b^*$. Intuitively, this achieves that the voters with unbounded residual should only receive as much residual as they need to match the maximum budget that any bounded voter has.

However, we find that this approach exhibits two problems. 
First, note that increasing residual when possible always improves the objective. As a consequence, \texttt{Convex Program} violates laminar-proportional-uniformity, since the 1-stability constraint allows voters from different ``subparties'' to receive different residuals depending on their size.
Second, the ``trade-off'' in the objective between the residuals and payments assigned to the voters approving some candidate $c\in C\setminus W$ (subject to the 1-stability constraint for $c$) may lead to unintuitive, undesirable behavior. For example, even when a committee $W$ provides perfect coverage, a voter $i$ who is part of many such 1-stability constraints may have a very low total payment, since increasing their payment necessitates reducing the residuals of many other voters to satisfy the constraints.  Hence, \texttt{Convex Program} violates perfect-coverage-uniformity.
Exploring suitable ways to overcome the challenges outlined above, or exploring alternative avenues for optimization-based approaches is an interesting direction for future work. 
\section{Additional Experimental Results}\label{app:exp}
In addition to the results for the Euclidean-VCR model and committee size $k=\lfloor m/2 \rfloor$ discussed in the main body, we report additional results in this section: We vary the committee size using $k\in\{\lfloor m/8 \rfloor,\,\lfloor m/4 \rfloor\}$.
Additionally, we evaluate our approach on 1000 profiles sampled from the resampling model, as well as for real-world participatory budgeting data from Pabulib.
\begin{itemize}
    \item \textbf{Pabulib:} We use real-world participatory budgeting election data from the Pabulib dataset \citep{Faliszewski23}. We collect all instances with approval ballots available as of January 23, 2026, with $m>4$ and average ballot length at least $3$. For instances with $n>500$, we subsample voters uniformly at random to $n=500$. Project costs are discarded.
    \item \textbf{Resampling:} 
    For each profile, we sample $n,m \in [10, 100]$ uniformly at random.
    We generate profiles using the resampling model \citep{DBLP:conf/ijcai/BoehmerFJ0LPRSS24} with parameters $\phi, p \in [0,1]$ sampled uniformly at random. A central ballot is sampled with each candidate approved with probability $p$. Each voter's ballot is then generated independently starting from the central ballot: for each candidate, with probability $1-\phi$ we keep the (dis-)approval for the candidate as in the central ballot, and with probability $\phi$ the candidate is resampled (approved independently with probability $p$). 
\end{itemize}
\subsection{EJR+ and Budgets}
We begin by reporting additional results for the experiments from \Cref{sec:exp:ejr}, where, for a given profile $A$ and committee $W$,  we compare the minimum $\alpha$ value for $\alpha$-EJR+ satisfied by $W$ and the minimum budget of a voter in the computed price system as a fraction of $\nicefrac{\lvert W\rvert}{n}$. The results are displayed in \Cref{fig:ejr:k-one}, \Cref{fig:ejr:k-two}, and \Cref{fig:ejr:k-three} for $k\in\{\lfloor m/8 \rfloor,\,\lfloor m/4 \rfloor,\,\lfloor m/2 \rfloor\}$, respectively. Additionally, we present histograms of the number of instances with $1$-EJR+ violations across the minimum budget fractions in \Cref{fig:ejr:k-hist-one}, \Cref{fig:ejr:k-hist-two}, and \Cref{fig:ejr:k-hist-three}.
\begin{figure}[htbp]
\centering
\includegraphics[width=\textwidth]{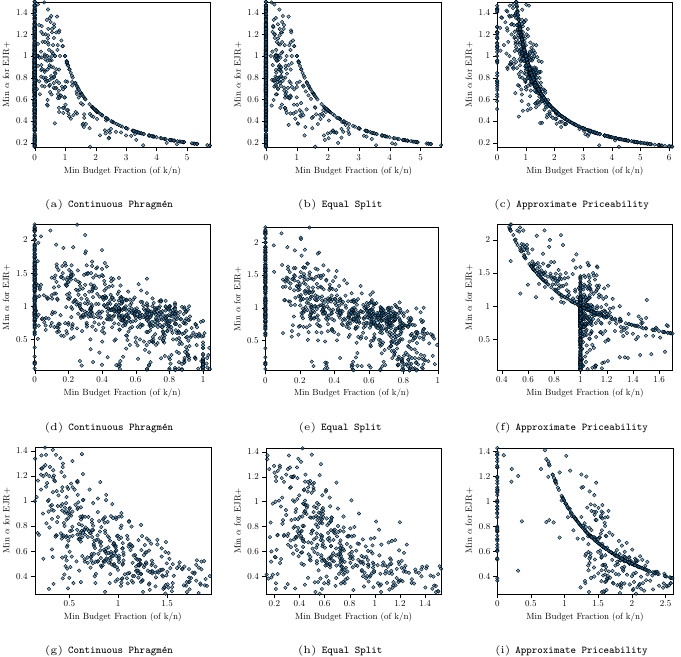}
\caption{Results for committees of size $k=\lfloor m/8 \rfloor$. Scatterplot of minimum budget and $\alpha$-EJR+ approximation across different datasets and explanation rules. Axes are clipped to the 5th–95th percentile range. Rows show different datasets: Euclidean (top), Resampling (middle), Pabulib (bottom).}
\label{fig:ejr:k-one}

\end{figure}

\begin{figure}[htbp]
\centering
\includegraphics[width=\textwidth]{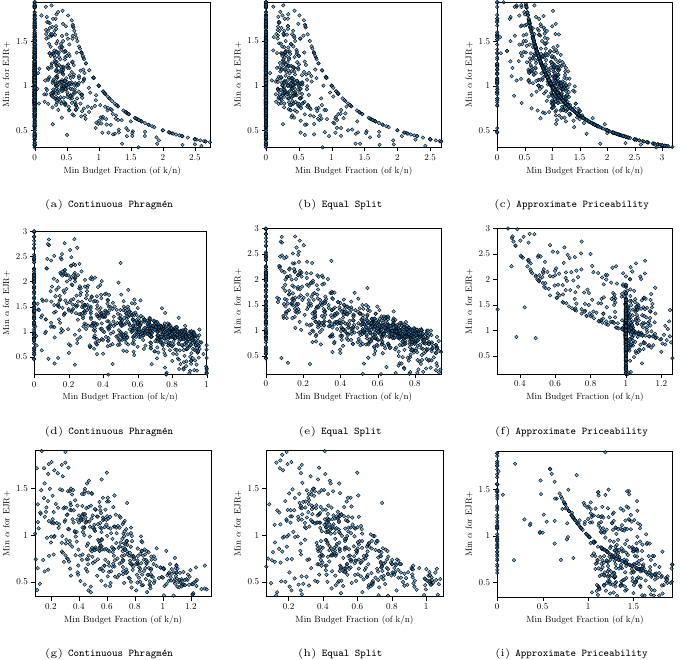}
\caption{Results for committees of size $k=\lfloor m/4 \rfloor$. Scatterplot of minimum budget and $\alpha$-EJR+ approximation across different datasets and explanation rules. Axes are clipped to the 5th–95th percentile range. Rows show different datasets: Euclidean (top), Resampling (middle), Pabulib (bottom).}
\label{fig:ejr:k-two}
\end{figure}
\begin{figure}[htbp]
\centering
\includegraphics[width=\textwidth]{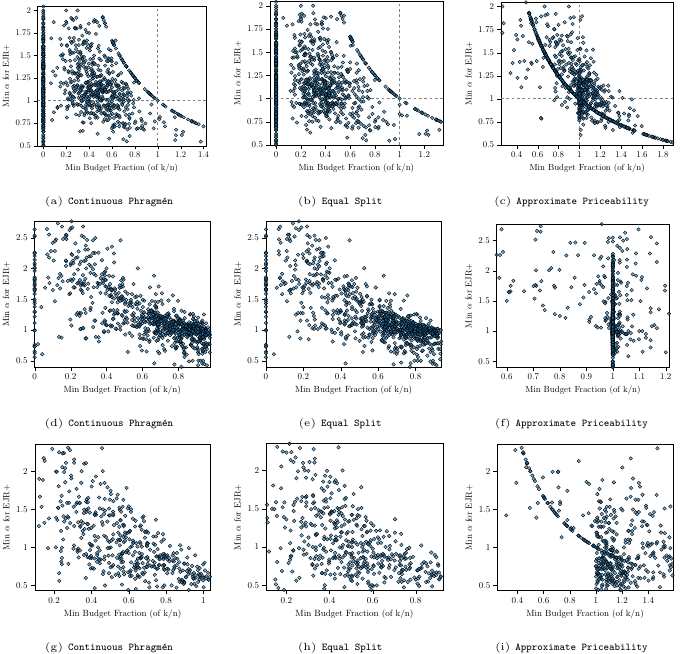}
\caption{Results for committees of size $k=\lfloor m/2 \rfloor$. Scatterplot of minimum budget and $\alpha$-EJR+ approximation across different datasets and explanation rules. Axes are clipped to the 5th–95th percentile range. Rows show different datasets: Euclidean (top), Resampling (middle), Pabulib (bottom). We also display the results for the Euclidean model already reported in the main body for completeness.}
\label{fig:ejr:k-three}
\end{figure}

\begin{figure}[htbp]
\centering
\includegraphics[width=\textwidth]{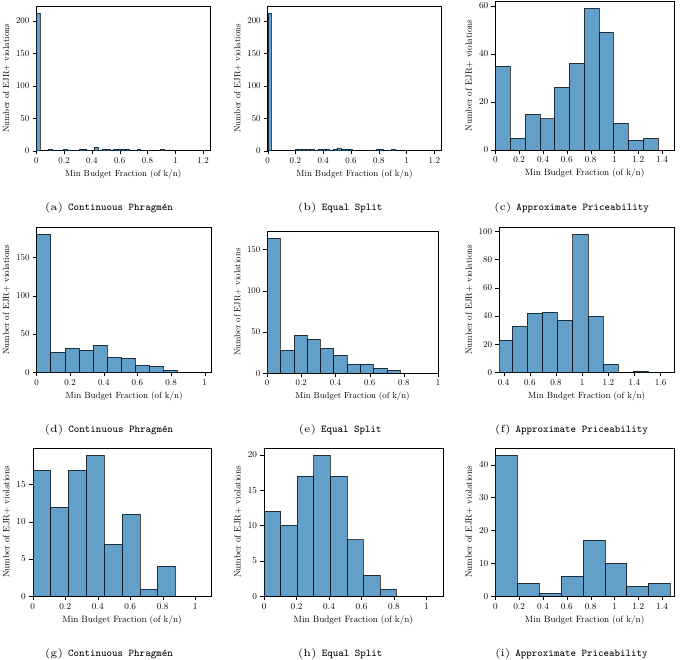}
\caption{Histogram of the number of instances with $1$-EJR+ violations across the minimum budget fractions. Results for committees of size $k=\lfloor m/8 \rfloor$ and different explanation rules.  Rows show different datasets: Euclidean (top), Resampling (middle), Pabulib (bottom).}
\label{fig:ejr:k-hist-one}
\end{figure}

\begin{figure}[htbp]
\centering
\includegraphics[width=\textwidth]{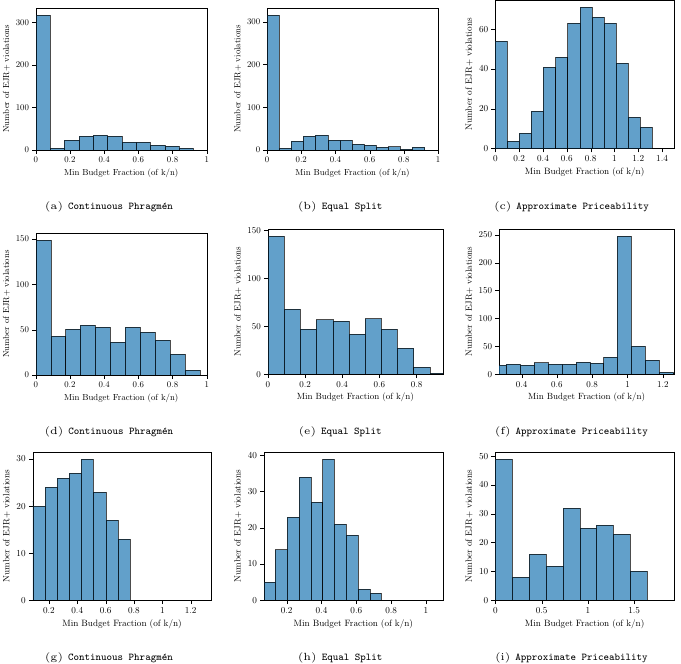}
\caption{Histogram of the number of instances with $1$-EJR+ violations across the minimum budget fractions. Results for committees of size $k=\lfloor m/4 \rfloor$ and different explanation rules.  Rows show different datasets: Euclidean (top), Resampling (middle), Pabulib (bottom).}
\label{fig:ejr:k-hist-two}
\end{figure}

\begin{figure}[htbp]
\centering
\includegraphics[width=\textwidth]{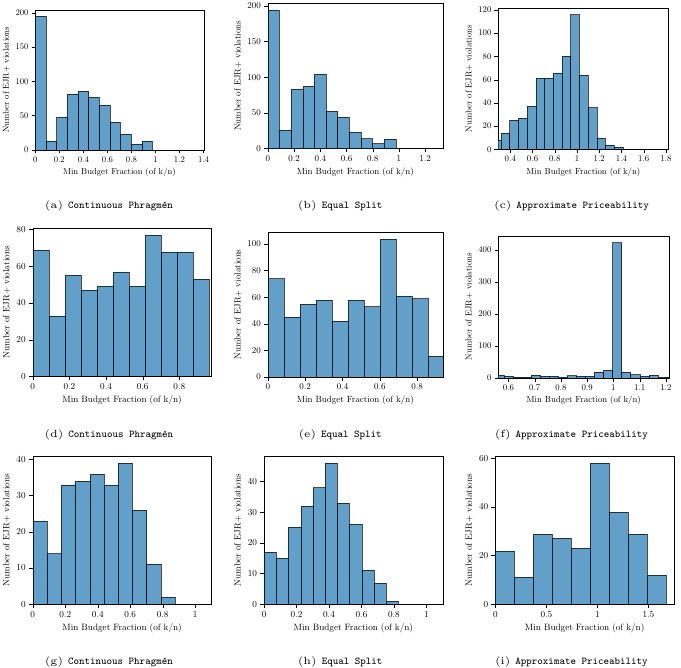}
\caption{Histogram of the number of instances with $1$-EJR+ violations across the minimum budget fractions. Results for committees of size $k=\lfloor m/2 \rfloor$ and different explanation rules.  Rows show different datasets: Euclidean (top), Resampling (middle), Pabulib (bottom).}
\label{fig:ejr:k-hist-three}
\end{figure}
We find that the observation made in the main body—that violations of EJR+ correspond with low minimum budgets, and that this relation is stronger for \texttt{Continuous Phragm\'{e}n} than for \texttt{Approximate Priceability}—is generally consistent across all datasets and committee sizes.

Comparing between datasets, we observe that for the resampling model, the differences between the price systems computed by \texttt{Continuous Phragm\'{e}n} or \texttt{Equal Split} compared to \texttt{Approximate Priceability} are more pronounced. 
Here, for \texttt{Continuous Phragm\'{e}n} and \texttt{Equal Split} the minimum budget fraction is below 1 on most instances. For \texttt{Approximate Priceability}, there is again a considerable fraction of the instances that have a minimum budget fraction close to 1 (this effect is especially apparent for $k=\lfloor m/2 \rfloor$), including instances that violate EJR+.  
For the Pabulib dataset, the results are similar to the resampling model: the range of minimum budget fractions is considerably smaller than for the Euclidean-VCR model, though generally slightly larger than for the resampling model. Moreover, we observe that a smaller number of the random committees violate $1$-EJR+ compared to the other datasets.

Comparing between committee sizes, we find that for larger values of $k$, the minimum budget fraction is generally considerably lower. A possible explanation for this effect might be that the budget assigned to underrepresented groups (i.e., a subset of voters who cannot receive a large budget due to the stability constraints) is lower in relative terms when the average budget is higher. The general relationship between minimum budget fractions and the minimum $\alpha$-EJR+ approximation is similar across committee sizes. 
We chose to report the results for $k=\lfloor m/2 \rfloor$ in the main body, as committees exhibit a wide range of $\alpha$-EJR+ approximations across datasets for this value. 

\subsection{Recovering Voting Power}
We now present the additional results for the experiments from \Cref{sec:exp:power}. As described in the main body, we consider instances where we generate voter weights on Euclidean-VCR instances using spatial weighting, and then compute the committee $W$ using the method of equal shares in the resulting weighted instance.
For each instance and each computed price system, we measure the Pearson correlation coefficient (PCC) between the ground-truth voter weights and the budgets output by the explanation rule.
In \Cref{tab:correlation_weights}, we report the additional results for a wider range of committee sizes $k\in\{\lfloor m/8 \rfloor,\,\lfloor m/4 \rfloor,\,\lfloor m/2 \rfloor\}$.

For $k=\lfloor m/4 \rfloor$, the results for \texttt{Continuous Phragm\'{e}n} and \texttt{Equal Split} are generally comparable to $k=\lfloor m/2 \rfloor$, though both methods exhibit a larger fraction of instances with a very strong correlation. For  \texttt{Approximate Priceability}, the correlation is stronger compared to $k=\lfloor m/2 \rfloor$.

For $k=\lfloor m/8 \rfloor$, the differences between the explanation rules are less pronounced. Moreover, for both \texttt{Continuous Phragm\'{e}n} and \texttt{Equal Split}, a smaller fraction of instances exhibits strong or very strong correlations.
A possible explanation for this is that for low committee sizes $k$, groups of voters with high weight can receive a lower number of additional selected candidates compared to larger values of $k$, limiting the degree to which they are overrepresented in the committee.

\begin{table}[htbp]
\centering
\small
\setlength{\tabcolsep}{4pt}
\renewcommand{\arraystretch}{1.1}
\caption{
Results on recovering voter weights obtained via spatial weighting on Euclidean-VCR instances through budgets returned by the explanation rule.
Percentage of instances with moderate (PCC $> 0.4$), strong (PCC $> 0.7$), and very strong (PCC $> 0.9$) correlation between voter weights and budgets for different committee sizes and explanation rules. We include the results for $k=\lfloor m/2 \rfloor$ already reported in the main body for completeness.}
\label{tab:correlation_weights}
\begin{tabular}{@{}l ccc ccc ccc@{}}
\toprule
& \multicolumn{3}{c}{\texttt{Continuous Phragm\'{e}n}} 
& \multicolumn{3}{c}{\texttt{Equal Split}} 
& \multicolumn{3}{@{}c@{}}{\texttt{Approx.\ Priceability}} \\
\cmidrule(lr){2-4} \cmidrule(lr){5-7} \cmidrule(lr){8-10}
$k$ 
& $>0.9$ & $>0.7$ & $>0.4$
& $>0.9$ & $>0.7$ & $>0.4$
& $>0.9$ & $>0.7$ & $>0.4$ \\
\midrule
$\lfloor m/8\rfloor$ 
& 9.7 & 47.6 & 81.1
& 16.5 & 51.1 & 81.7
& 10.0 & 48.0 & 82.5 \\
$\lfloor m/4 \rfloor$ 
& 18.2 & 64.3 & 89.9
& 33.2 & 70.0 & 90.7
& 11.6 & 57.1 & 87.0 \\
$\lfloor m/2 \rfloor$ 
& 11.9 & 65.3 & 91.1
& 28.1 & 72.8 & 92.9
& 3.3  & 37.3 & 79.3 \\
\bottomrule
\end{tabular}
\end{table}

\subsection{Experiments on Detecting Bias in Text Summarization}\label{app:exp:textsum}
We sample $n=75$ text statements from the Bowling Green dataset available from \citet{polis_data_2026}, where citizens comment on possible improvements to the city of Bowling Green in KY, US.
We manually extract six common topics: Infrastructure \& Transit, Social Equity, Urban Planning \& Environment, Public Health \& Safety, Education \& Youth, and Economic Development.
Then, we assign each statement at most one topic (some statements may remain unlabeled). The labeled dataset is available on \href{https://github.com/luca-kreisel/Explanation-Systems}{GitHub}.

To construct an approval profile, we embed the text statements using the LLM-based \texttt{text-embedding-3-large} (OpenAI). 
We compute pairwise cosine similarities between all statement embeddings and set a threshold at the 80th percentile of these similarities. Cosine similarities have been observed to correlate with human judgments \citep{DBLP:conf/emnlp/GaoYC21}, and have been used in recent work by \citet{DGG+26a} to audit representation in online deliberative processes.
Each statement (i.e., voter) then approves all statements (i.e., the respective candidates, including the candidate corresponding to their own statement) whose similarity exceeds this threshold. We then select a subset of $k=10$ statements as a summary (i.e., a committee $W$) using the LSA summarization method\footnote{Note that the summarization algorithm is independent of the assigned labels, generated embeddings, or approvals.} \citep{DBLP:journals/jasis/DeerwesterDLFH90}.

Based on the manually assigned labels, we find that the produced summary seems to provide an unbalanced representation of the input statements. Specifically, no statement assigned to Infrastructure \& Transit (21\% of statements) is selected, while statements concerning Urban Planning \& Environment (50\% of the summary; 16\% in the input) and Education \& Youth (30\% of the summary; 9\% in the input) are disproportionally represented. 

We then compute price systems for $W$ and the approval profile using  \texttt{Approximate Priceability}, \texttt{Continuous Phragm\'{e}n}, and \texttt{Equal Split}. For \texttt{Approximate Priceability}, despite the apparent disproportional representation, we again observe a close to budget-uniform price system: all voters have a budget of at least $0.18$, and all but 5 voters have a budget of at most $0.25$. In contrast, for \texttt{Continuous Phragm\'{e}n} and \texttt{Equal Split}, we observe lower minimum budgets that reflect the disproportional representation. For \texttt{Continuous Phragm\'{e}n}, there are 40 voters with budget below $k/n\approx 0.13$, for \texttt{Equal Split} there are 35 such voters.
Moreover, the average budgets across topics reflect their level of representation: the overrepresented topics (Urban Planning \& Environment, Education \& Youth) have the highest average budgets ($0.27$ and $0.29$ for \texttt{Equal Split}; $0.25$ and $0.36$ for \texttt{Continuous Phragmén}), while the underrepresented topic Infrastructure \& Transit has the lowest average budget ($0.12$ for both \texttt{Equal Split} and \texttt{Continuous Phragmén}). 
Looking more closely at the statements with budget below $k/n$, we find that their budgets are mostly consistent with their representation. As an example, the budgets point to statements that advocate for improving traffic flow as underrepresented (i.e., the voters corresponding to these statements have low budgets), which is consistent with these opinions being present in the input but not substantially reflected in the summary. Similarly, several statements calling for improved internet and cable TV services have low budgets.

These results indicate that explanations computed by \texttt{Continuous Phragmén} and \texttt{Equal Split} can be a useful tool to detect and analyze proportional representation in subset selection problems, and that price systems can be used to detect unfairness at a more fine-grained level.

\paragraph{Limitations}
We remark on the following limitations of our exploratory analysis. Constructing approval profiles from text embeddings is non-trivial and involves several design choices (e.g., the embedding model and threshold). While the approvals generated by our approach seem to largely be sensible, we observe some outliers.  Moreover, more sophisticated approaches beyond setting approvals based on a threshold are possible.
Investigating this challenge further and validating our results on more datasets is an interesting opportunity for future research.

\end{document}